\documentclass[a4paper,UKenglish,cleveref, autoref, thm-restate]{lipics-v2021}

\pdfoutput=1 
\hideLIPIcs  

\usepackage{amsmath, amsthm, tikz, tikz-cd}
\newtheorem{lem}{Lemma}
\newtheorem{defn}[lem]{Definition}
\newtheorem{theor}[lem]{Theorem}
\newtheorem{prop}[lem]{Proposition}
\newtheorem{fact}[lem]{Fact}
\newtheorem{obs}[lem]{Observation}

\newtheorem{quest}[lem]{Question}
\usepackage{mathtools}
\usepackage{enumitem}
\newcommand{\defeq}{\vcentcolon=}
\usepackage{xcolor}

\newcommand{\abs}[1]{|#1|}
\newcommand{\cat}[1]{\mathbf{#1}}
\newcommand{\logic}[1]{#1}
\newcommand{\sheaf}[1]{\mathcal{#1}}
\newcommand{\bigO}{\mathcal{O}}

\title{Cohomology in Constraint Satisfaction and Structure Isomorphism} 

\author{Adam {Ó Conghaile}}{Computer Laboratory, Cambridge University, United Kingdom \and The Alan Turing Institute, United Kingdom \and \url{aoconghaile.github.io} }{ac891@cam.ac.uk}{https://orcid.org/
0000-0002-3032-5514}{}

\authorrunning{A. Ó Conghaile} 

\Copyright{Adam Ó Conghaile} 

\ccsdesc[500]{Theory of computation~Finite Model Theory}
\keywords{constraint satisfaction problems, finite model theory, descriptive complexity, rank logic, Weisfeiler-Leman algorithm, cohomology}

\category{} 

\relatedversion{} 

\acknowledgements{I thank Anuj Dawar and Samson Abramsky for useful discussions in
writing this paper and am especially gratefully to Samson Abramsky for permission to
use results from \cite{samsonnote}.}

\nolinenumbers 

\EventEditors{John Q. Open and Joan R. Access}
\EventNoEds{2}
\EventLongTitle{42nd Conference on Very Important Topics (CVIT 2016)}
\EventShortTitle{CVIT 2016}
\EventAcronym{CVIT}
\EventYear{2016}
\EventDate{December 24--27, 2016}
\EventLocation{Little Whinging, United Kingdom}
\EventLogo{}
\SeriesVolume{42}
\ArticleNo{23}

\begin{document}

\maketitle

\begin{abstract}
%

Constraint satisfaction (CSP) and structure isomorphism (SI) are among the most well-studied computational problems in Computer Science. While neither problem is thought to be in \texttt{PTIME}, much work is done on \texttt{PTIME} approximations to both problems. Two such historically important approximations are the $k$-consistency algorithm for CSP and the $k$-Weisfeiler-Leman algorithm for SI, both of which are based on propagating local partial solutions. The limitations of these algorithms are well-known – $k$-consistency can solve precisely those CSPs of bounded width and $k$-Weisfeiler-Leman can only distinguish structures which differ on properties definable in $C^k$.
In this paper, we introduce a novel sheaf-theoretic approach to CSP and SI and their approximations. We show that both problems can be viewed as deciding the existence of global sections of presheaves, $\sheaf{H}_k(A,B)$ and $\sheaf{I}_k(A,B)$ and that the success of the $k$-consistency and $k$-Weisfeiler-Leman algorithms correspond to the existence of certain efficiently computable subpresheaves of these. Furthermore, building on work of Abramsky and others in quantum foundations, we show how to use Čech cohomology in $\sheaf{H}_k(A,B)$ and $\sheaf{I}_k(A,B)$ to detect obstructions to the existence of the desired global sections and derive new efficient cohomological algorithms extending $k$-consistency and $k$-Weisfeiler-Leman. We show that cohomological $k$-consistency can solve systems of
equations over all finite rings and that cohomological Weisfeiler-Leman can
distinguish positive and negative instances of the Cai-F\"urer-Immerman property
over several important classes of structures.

\end{abstract}

\section{Introduction}

Constraint satisfaction problems (CSP) and structure isomorphism (SI) are two of the most
well-studied problems in complexity theory. Mathematically speaking, an instance of
one of these problems takes a pair of structures $(A,B)$ as input and asks
whether there is a homomorphism
$A \rightarrow B$ for CSP or an isomorphism $A \cong B$ for
SI. These problems are not in general thought to be tractable. Indeed the general
case of CSP is \texttt{NP-Complete} and restricting our structures to graphs the
best known algorithm for SI is Babai's quasi-polynomial time algorithm~\cite{babai}.
As a result, it is common in complexity and finite model theory to study approximations
of the relations $\rightarrow$ and $\cong$.

The $k$-consistency and $k$-Weisfeiler-Leman\footnote{The algorithm we call ``$k$-Weisfeiler-Leman'' is more commonly called ``$(k-1)$-Weisfeiler-Leman'' in the literature, see for example \cite{cfi}. We prefer ``$k$-Weisfeiler-Leman'' to emphasise its relationship to $k$-variable logic and sets of $k$-local isomorphisms.} algorithms efficiently determine two such approximations
to $\rightarrow$ and $\cong$ which we call $\rightarrow_k$ and $\equiv_k$. These relations have many characterisations
in logic and finite model theory, for example in \cite{federvardi} and \cite{cfi}.
One that is particularly useful is that of the existence of
winning strategies for Duplicator in certain Spoiler-Duplicator games with $k$ pebbles~\cite{Kolaitis1995, IL}.
For both of these games Duplicator's winning strategies can be represented as non-empty
sets $S \subset \mathbf{Hom}_k(A,B)$ of $k$-local partial homomorphisms which satisfy
some extension properties and connections between these games have been studied before.
For example, a joint comonadic semantics is given by the pebbling comonad of Abramsky,
Dawar and Wang~\cite{pebb}.

The limitations of these approximations are well-known.
In particular, it is known that $k$-consistency only solves CSPs of \emph{bounded width} and
$k$-Weisfeiler-Leman can only distinguish structures which differ on properties expressible in the infinitary counting
logic $\logic{C}^k$. Feder and Vardi~\cite{federvardi} showed that CSP encoding linear
equations over the finite fields do not have bounded width, while Cai, F\"urer, and Immerman~\cite{cfi}
demonstrated an efficiently decidable graph property which is not expressible in
 $\logic{C}^k$ for any $k$.

In the present paper, we introduce a novel approach to the CSP and SI problems
based on presheaves of $k$-local partial homomorphisms and isomorphisms, showing
that the problems can be reframed as deciding whether certain presheaves admit global sections.
We show that the classic $k$-consistency and $k$-Weisfeiler-Leman algorithms can
be derived by computing greatest fixpoints of presheaf operators which remove some
efficiently computable obstacles to global sections. Furthermore, we show how
invariants from sheaf cohomology can be used to find further obstacles to combining
local homomorphisms and isomorphisms into global ones. We use these to construct new efficient extensions
to the $k$-consistency and $k$-Weisfeiler-Leman algorithms computing relations
$\rightarrow_k^{\mathbb{Z}}$ and $\equiv^{\mathbb{Z}}_k$ which refine $\rightarrow_k$ and $\equiv_k$.

The application of presheaves has been particularly successful in computer science
in recent decades with applications in semantics~\cite{staton, statonfiori}, information theory~\cite{bennequin}
and quantum contextuality~\cite{ab_brand, ab_man_barb, abramsky_et_al}.
This work draws in particular on the application
of sheaf theory to quantum contextuality, pioneered by Abramsky and Brandenburger~\cite{ab_brand}
and developed by Abramsky and others for example in \cite{ab_man_barb} and \cite{abramsky_et_al}.

Using this work, we prove that these new cohomological
algorithms are strictly stronger than $k$-consistency and $k$-Weisfeiler-Leman. In
particular, we show that cohomological $k$-consistency decides solvability of linear equations
with $k$ variables per equation over all finite rings and that there is a fixed $k$
such that $\equiv^{\mathbb{Z}}_k$ distinguishes structures which differ on Cai, F\"urer and Immerman's property.

It is also interesting to compare $\rightarrow^{\mathbb{Z}}_k$ and $\equiv^{\mathbb{Z}}_k$ with
other well-studied refinements of $\rightarrow_k$ and $\equiv_k$. For $\rightarrow_k$,
such refinements include the algorithms of Bulatov~\cite{Bulatov} and Zhuk~\cite{Zhuk}
which decide all tractable CSPs and the algorithms of Brakensiek, Guruswami, Wrochna
and \v{Z}ivn\'{y}~\cite{aip} and Ciardo and \v{Z}ivn\'{y}~\cite{clap}
for Promise CSPs. For $\equiv_k$, comparable approximations to $\cong$ include linear
Diophantine equation methods employed by Berkholz and Grohe~\cite{BerkholzGrohe}
and the invertible-map equivalence of Dawar and Holm~\cite{anuj_bjarki} which bounds
the expressive power of rank logic. The latter was recently used by Lichter~\cite{lich} to demonstrate a property which is decidable in \texttt{PTIME} but not expressible in rank logic. In our paper, we show that
$\equiv^{\mathbb{Z}}_k$, for some fixed $k$, can distinguish structures which differ on this property.
Comparing $\rightarrow^{\mathbb{Z}}_k$ to the Bulatov-Zhuk algorithm and algorithms for PCSPs remains a direction
for future work.

The rest of the paper proceeds as follows. Section \ref{sec:back} establishes some
background and notation.
Section \ref{sec:presheaf} introduces the presheaf formulation of CSP and SI and
new formulations of $k$-consistency and $k$-Weisfeiler-Leman in this framework.
Section \ref{sec:cohom} demonstrates how to apply aspects of sheaf cohomology to CSP and
SI and defines new algorithms along these lines.
Section \ref{sec:effect} surveys the strength of these new cohomological algorithms.
Section \ref{sec:conclusion} concludes
with some open questions and directions for future work.
Major proofs and additional background are left to the appendices.

\section{Background and definitions}\label{sec:back}


In this section, we record some definitions and background which are necessary for
our work.

\subsection{Relational structures \& finite model theory}
Throughout this paper we use the word \emph{structure} to mean a relational
structure over some finite relational signature $\sigma$. A structure $A$ consists of
an underlying set (which we also call $A$) and for each relational symbol $R$
of arity $r$ in $\sigma$ a subset $R^A \subset A^{r}$ or tuples related by $R$.
A \emph{homomorphism} of structures $A, B$ over a common signature is a function  between
the underlying sets $f \colon A \rightarrow B$ which preserves related tuples. An
\emph{isomorphism} of structures is a bijection between the underlying sets which
both preserves and reflects related tuples. A partial function $s \colon A \rightharpoonup B$
(seen as a set $s \subset A\times B$) is a \emph{partial homomorphisms} if it
preserves the related tuples in $\mathbf{dom}(s)$.
$s$ is a \emph{partial isomorphism} if it is a bijection onto its image and both
preserves and reflects related tuples. A partial homomorphism or isomorphism is
said to be \emph{$k$-local} if $\abs{\mathbf{dom}(s)} \leq k$. For two structures
over the same signature we write $\mathbf{Hom}_k(A,B)$ and $\mathbf{Isom}_k(A,B)$
respectively for the sets of $k$-local homomorphisms and isomorphisms from
$A$ to $B$.

In the paper, we make reference to several important logics from finite model theory
and descriptive complexity theory. The logics we make reference to in this paper are as follows.
\begin{itemize}
  \item Fixed-point logic with
  counting (written FPC) is first-order logic extended with operators for inflationary fixed-points and
  counting, for example see \cite{fpc}.
  \item For any natural number $k$, $\logic{C}^k$ is infinitary first-order logic extended with
  counting quantifiers with at most $k$ variables. This logic bounds the expressive power of FPC in the sense that, for each $k'$ there exists $k$ such that any FPC formula in $k'$ variables is equivalent to one in $\logic{C}^{k}$. We write $\logic{C}^{\omega}$ for the union of these logics.
  \item Rank logic is first-order logic extended with operators for inflationary fixed-points and
  computing ranks of matrices over finite fields, see \cite{pakusa}.
  \item Linear algebraic logic is first-order infinitary logic extended with quantifiers for computing \textit{all} linear algebraic functions over finite fields, see \cite{dgp}. This logic bounds rank logic in the sense described above.
\end{itemize}

At different points in the history of descriptive complexity theory,
both FPC and rank logic were considered as candidates for ``capturing \texttt{PTIME}''
and thus refuting a well-known conjecture of Gurevich~\cite{gure}. Each has since been proven not
to capture \texttt{PTIME}, for FPC see Cai, F\"urer and Immerman~\cite{cfi}, for rank logic
see Lichter~\cite{lich}. Infinitary logics such as $\logic{C}^{\omega}$ and linear algebraic logic are capable of expressing properties which are not decidable in \texttt{PTIME} but have been shown not to contain any logic which does not capture \texttt{PTIME}. For $\logic{C}^{\omega}$, see Cai, F\"urer and Immerman~\cite{cfi} and for linear algebraic logic, see Dawar, Gr\"adel, and Lichter~\cite{dgl}.

\subsection{Constraint satisfaction problems \& Structure Isomorphism}
Assuming a fixed relational signature $\sigma$, we write $CSP$ for the
set of all pairs of $\sigma$-structures $(A,B)$ such that there is a homomorphism
witnessing $A \rightarrow B$. We use $CSP(B)$ to denote the set of
relational structures $A$ such that $(A,B) \in CSP$. We also use $CSP$ and $CSP(B)$
to denote the decision problem on these sets. For general $B$, $CSP(B)$ is well-known to be
\texttt{NP-complete}. However for certain structures $B$ the problem is in \texttt{PTIME}.
Indeed, the Bulatov-Zhuk Dichotomy Theorem (formerly the Feder-Vardi Dichotomy
Conjecture) states that, for any $B$, $CSP(B)$ is either \texttt{NP}-complete
or it is \texttt{PTIME}. Working out efficient algorithms which decide
$CSP(B)$ for larger and larger classes of $B$ was an active area of research
which culminated in Bulatov and Zhuk's exhaustive classes of algorithms \cite{Bulatov,Zhuk}.

Similarly, we write $SI$ for the
set of all pairs of $\sigma$-structures $(A,B)$ such that there is an isomorphism
witnessing $A \cong B$. The decision problem for this set is also thought not to be
in \texttt{PTIME} however there are no general hardness results known for this. The
best known algorithm (in the case where $\sigma$ is the signature of graphs) is
Babai's~\cite{babai} which is quasi-polynomial.

There are many efficient algorithms which approximate the decision problems of
$CSP$ and $SI$. Two such examples, which are of particular importance to this paper,
are the $k$-consistency and $k$-Weisfeiler-Leman algorithms. Explicit modern presentations of
these algorithms can be seen, for example, in \cite{AtseriasBulatovDalmau} and \cite{KieferThesis}.
We instead focus on equivalent formulations in terms of positional Duplicator winning strategies.
These are given by Kolaitis and Vardi~\cite{Kolaitis1995} for $k$-consistency and Hella~\cite{lhipt}
for $k$-Weisfeiler-Leman. In the case of $k$-consistency,
a pair $(A,B)$ is accepted by the algorithm if and only if there is a non-empty subset
$S \subset \mathbf{Hom}_k(A,B)$ which is downward-closed and satisfies the so-called forth
property. This means any $s \in S$ with $\abs{\mathbf{dom}(s)}< k$ satisfies the property
$\mathbf{Forth}(S,s)$ which is defined as  $$\forall a \in A, \
        \exists b \in B \text{ s.t. } s \cup \{(a,b)\} \in S.$$ If such an $S$ exists
we write $A \rightarrow_k B$. The similar strategy-based characterisation of $k$-Weisfeiler-Leman is
captured by non-empty downward-closed sets $S \subset \mathbf{Isom}_k(A,B)$ where each
element satisfies the bijective forth property $ \mathbf{BijForth}(S,s)$ defined by
$$\exists b_s\colon A \rightarrow B \text{ a bijection s.t. }
\forall a \in A \ s \cup \{(a,b_s(a)) \} \in S.$$
If such an $S$ exists
we write $A \equiv_k B$. For more details, see Appendix \ref{sec:app_algs}.

\subsection{Presheaves \& cohomology}

Here we give a brief account of the category-theoretic preliminaries for this paper.
For a more comprehensive introduction to category theory we refer to Chapter 1 of Leinster's textbook~\cite{Leinster}
and for a complete account of presheaves we refer to Chapter 2 of MacLane and Moerdijk~\cite{MacLane}.

Given two categories $\cat{C}$ and $\cat{S}$, an $\cat{S}$-valued presheaf over $\cat{C}$ is a contravariant
functor $\mathcal{F} \colon \cat{C}^{op} \rightarrow \cat{S}$. We will assume that
$\cat{C}$ is some subset of the powerset of some set $X$ with subset inclusion as
the morphisms. We call $X$ the underlying space of $\cat{C}$.
For this reason, when $U' \subset U$ in $\cat{C}$ we write $(\cdot)_{\mid_{U'}}$
for the restriction map $\sheaf{F}(U' \subset U)\colon \sheaf{F}(U) \rightarrow \sheaf{F}(U')$.
We assume $\cat{S}$ is either the category $\mathbf{Set}$ of sets or
the category $\mathbf{AbGrp}$ of abelian groups. We call $\mathbf{AbGrp}$-valued presheaves,
abelian presheaves. $\mathbf{Set}$-valued presheaves are just called presheaves or presheaves
of sets where there is ambiguity.

For any $\cat{C}$ and $\cat{S}$ as above, the category of presheaves $\cat{PrSh}(\cat{C}, \cat{S})$ has as objects the presheaves $\sheaf{F}\colon \cat{C}^{op} \rightarrow \cat{S}$ and, as morphisms, natural transformations between these functors. If $\cat{S}$ has a terminal object $1$ (as both $\cat{Set}$ and $\cat{AbGp}$ do) then the presheaf $\mathbb{I} \in \cat{PrSh}(\cat{C}, \cat{S})$ which sends all elements of $\cat{C}$ to $1$ is a terminal object in
$\cat{PrSh}(\cat{C}, \cat{S})$. For any $\sheaf{F} \in \cat{PrSh}(\cat{C}, \cat{S})$, a \emph{global section} of $\sheaf{F}$ is a natural transformation $S \colon \mathbb{I} \implies \sheaf{F}$.

\section{Presheaves of local homomorphisms and isomorphisms}\label{sec:presheaf}


Some important efficient algorithms for CSP and SI involve working with sets of $k$-local
homomorphisms between the two structures in a given instance. These sets of partial
homomorphisms of domain size $\leq k$ are useful for constructing efficient algorithms
because computing the sets $\mathbf{Hom}_k(A, B)$ and $\mathbf{Isom}_k(A,B)$ can be
done in polynomial time in $\abs{A}\cdot \abs{B}$. In this section, we see that
these sets can naturally be given the structure of sheaves, that the CSP and SI
problems can be seen as the search for global sections of these sheaves and that
the $k$-consistency and $k$-Weisfeiler-Leman algorithms can both be seen as
determining the existence of certain special subpresheaves. The framework of considering
sheaves of local homomorphisms and isomorphisms is novel in this work and essential
for the main cohomological algorithms later. The results in Section \ref{sec:samsons}
are from a technical report of Samson Abramsky~\cite{samsonnote} and we thank
him for his permission to include them here.

\subsection{Defining presheaves of homomorphisms and isomorphisms}
Let $A$ and $B$ be relational structures over the same signature.
A \emph{partial homomorphism} is a partial function $s \colon A \rightharpoonup B$ that preserves related tuples in $\mathbf{dom}(s)$.
A \emph{partial isomorphism} is a partial homomorphism $s \colon A \rightharpoonup B$
which is injective and reflects related tuples from $\mathbf{im}(s)$.
A \emph{$k$-local homomorphism (resp.\ isomorphism)} is a partial homomorphism (resp.\ isomorphism) $s$ such that $\abs{\mathbf{dom}(s)} \leq k$.
We write $\mathbf{Hom}_k(A,B)$ (resp.\ $\mathbf{Isom}_k(A,B)$) for the sets of $k$-local homomorphisms (resp.\ isomorphisms).
We write $\mathbf{Hom}(A,B)$ for the union $\bigcup_{1\leq k \leq \abs{A}} \mathbf{Hom}_k(A,B)$ and $\mathbf{Isom}(A,B)$ for the union $\bigcup_{1\leq k \leq \abs{A}} \mathbf{Isom}_k(A,B)$.\\

It is not hard to see that these sets can be given the structure of presheaves
on the underlying space $A$. Indeed, we define the \emph{presheaf of homomorphisms
from $A$ to $B$} $\sheaf{H}(A,B) \colon \cat{P(A)}^{op} \rightarrow \cat{Set}$ as
$\sheaf{H}(A,B)(U) = \{ s \in \mathbf{Hom}(A,B) \ \mid \ \mathbf{dom}(s) = U \}$
with restriction maps $\sheaf{H}(A,B)(U' \subset U)$ given by the restriction
of partial homomorphisms $(\cdot)_{\mid_{U'}}$. Similarly, let $\sheaf{I}(A,B)$
be the subpresheaf of $\sheaf{H}(A,B)$ containing only partial isomorphisms.
Now, consider the cover of $A$ by subsets of size at most $k$, written $A^{\leq k} \subset P(A)$.
We define the \emph{presheaves of $k$-local homomorphisms and isomorphisms}
$\sheaf{H}_k(A,B)$ and $\sheaf{I}_k(A,B)$ as the functors $\sheaf{H}(A,B)$ and $\sheaf{I}(A,B)$
restricted to the subcategory  $(\cat{A^{\leq k}})^{op} \subset \cat{P(A)}^{op}$.

We now see how these presheaves and their global sections encode the CSP and SI problems for the instance $(A,B)$.

\subsection{CSP and SI as search for global sections}

Fix an instance $(A,B)$ for the CSP or SI problem and let $\sheaf{H}$ and $\sheaf{I}$ stand for the presheaves of all partial homomorphisms and isomorphisms between $A$ and $B$ defined in the last section. For either of these presheaves $\sheaf{S}$ a global section $s \colon \mathbb{I} \implies \sheaf{S}$ is a collection $ \{ s_{U} \in \sheaf{S}(U) \}_{U \in P(A)}$ where naturality implies that for any subsets $U$ and $U'$ of $A$ $(s_U)_{\mid_{U\cap U'}} = (s_{U'})_{\mid_{U\cap U'}}$. As the poset $P(A)$ has a maximal element, namely $A$, any such global section is determined by a choice of $s_A \in \sheaf{S}(A)$. This leads us to the following observation.

\begin{obs}
  Given a pair $(A,B)$ relational structures over the same signature then
  $$ (A,B) \in CSP \iff \sheaf{H} \text{ has a global section} $$
  and if $\abs{A} = \abs{B}$ then
  $$ (A,B) \in SI \iff \sheaf{I} \text{ has a global section.} $$
\end{obs}

This observation reframes the CSP and SI problems in terms of presheaves but algorithmically this not a particularly useful restating as computing the full objects $\sheaf{H}$ and $\sheaf{I}$ requires solving the CSP and SI problems for all
subsets of $A$ and $B$. A much more interesting equivalent condition is that for large enough $k$, whether or not a particular instance $(A,B)$ is in CSP or SI is determined by the global sections of the presheaves of $k$-local homomorphisms and isomorphisms.

\begin{lem}\label{lem:klocal}
For a pair $(A,B)$ relational structures over the same signature, $\sigma$, and $k$ at least the arity of $sigma$ then
$$ (A,B) \in CSP \iff \sheaf{H}_k \text{ has a global section} $$
and if $\abs{A} = \abs{B}$ then
$$ (A,B) \in SI \iff \sheaf{I}_k \text{ has a global section.} $$
\end{lem}

\begin{proof}
  See Appendix \ref{sec:app0}.
\end{proof}

This is more interesting than the previous observation as $\sheaf{H}_k$ and $\sheaf{I}_k$ can be computed for any relational structures $A$ and $B$ in $\bigO(\text{poly}(\abs{A}\cdot\abs{B}))$. Indeed, we can just list all $\bigO(\abs{A}^k\cdot\abs{B}^k)$ possible $k$-local functions and check which ones preserve (and reflect) related tuples. This also gives us an interesting starting point for designing efficient algorithms for approximating CSP and SI. In particular, any efficient algorithms which finds obstacles to the existence of global sections in $\sheaf{H}_k$ and $\sheaf{I}_k$ will provide a tractable approximation to CSP and SI. We now see how this approach can be used to capture some classical approximations of these problems.

\subsection{Algorithms and games in terms of presheaves} \label{sec:samsons}

In this section, we consider the approximations $A \rightarrow_k B$ and $A \equiv_k B$ to CSP and SI which are computed respectively by the $k$-consistency and $k$-Weisfeiler-Leman algorithms and we show that these algorithms can be seen as searching for certain obstructions to global sections in $\sheaf{H}_k(A,B)$ and $\sheaf{I}_k(A,B)$. In particular, we define efficiently computable monotone operators on subpresheaves of $\sheaf{H}_k$ and $\sheaf{I}_k$ and show that they have non-empty greatest fixpoints if and only if $(A, B)$ are accepted by $k$-consistency and $k$-Weisfeiler-Leman respectively. Proposition \ref{prop:sam1} is reproduced with permission from an unpublished technical report of Samson Abramsky and the formulation of the fixpoint operators is inspired by the same report.

\subsubsection{Flasque presheaves and $k$-consistency}

Recall that $A \rightarrow_k B$ if and only if there is a positional winning strategy for Duplicator in the existential $k$-pebble game~\cite{federvardi} and
that a presheaf $\sheaf{F}$ is flasque if all of the restriction maps $\sheaf{F}(U \subset U')$ are surjective.
In a recent technical report, Abramsky~\cite{samsonnote} proves the following characterisation of these strategies in our presheaf setting.

\begin{prop}\label{prop:sam1}
  For $A, B$ relational structures and any $k$ there is a bijection between:
\begin{itemize}
  \item positional strategies in the existential $k$-pebble game from $A$ to $B$, and
  \item non-empty flasque subpresheaves $\sheaf{S} \subset \sheaf{H}_k(A,B)$.
\end{itemize}
\end{prop}

This gives an alternative description of the $k$-consistency algorithm as constructing
the largest flasque subpresheaf $\overline{\sheaf{H}_k}$ of $\sheaf{H}_k$ and
checking if it is empty. As pointed out by Abramsky~\cite{samsonnote}, this is the
process of coflasquification of the presheaf $\sheaf{H}_k$ and can be seen as
dual to the Godement construction~\cite{Godement}, an important early construction
in  homological algebra. $\overline{\sheaf{H}_k}$ can be computed efficiently as the greatest fixpoint
of the presheaf operator $(\cdot)^{\uparrow\downarrow}$ which computes the largest
subpresheaf of a presheaf $\sheaf{S} \subset \sheaf{H}_k$ such that every $s \in \sheaf{S}^{\uparrow\downarrow}(C)$
satisfies the forth property $\mathbf{Forth}(\sheaf{S}, s)$. For further details see
Appendix \ref{sec:app_algs}

\subsubsection{Greatest fixpoints and $k$-Weisfeiler-Leman}

In a similar way to the $k$-consistency algorithm, $k$-Weisfeiler-Leman can be
formulated as determining the existence of a positional strategy for Duplicator in the
$k$-pebble bijection game between $A$ and $B$. This inspires the definition of
another efficiently computable presheaf operator $(\cdot)^{\# \downarrow}$ which computes the largest
subpresheaf of a presheaf $\sheaf{S} \subset \sheaf{I}_k$ such that for every $s \in \sheaf{S}^{\# \downarrow}(C)$
satisfies the bijective forth property $\mathbf{BijForth}(\sheaf{S}, s)$.  We call
the greatest fixpoint of this operator $\overline{\overline{\sheaf{S}}}$ and we
have that $A \equiv_k B$ if and only if $\overline{\overline{\sheaf{I}_k}}$ is non-empty.
For more details, see Appendix \ref{sec:app_algs}.

To conclude, in this section, we have seen how to reformulate the search for
homomorphisms and isomorphisms between relational structures $A$ and $B$ as the search for
global sections in the presheaves $\sheaf{H}_k(A,B)$ and $\sheaf{I}_k(A,B)$. We have
also seen that well-known approximations of homomorphism and isomorphism, $\rightarrow_k$ and
$\equiv_k$, can be computed as greatest fixpoints of presheaf operators which remove elements
which cannot form part of any global section. In the next section, we look at sheaf-theoretic
obstructions to forming a global section which come from cohomology and see how these
can be used to define stronger approximations of homomorphism and isomorphism.

\section{Cohomology of local homomorphisms and isomorphisms}\label{sec:cohom}


As we showed in the previous section, an instance
of CSP and SI with input $(A,B)$ can be seen as determining the existence
of a global section for the presheaf $\mathcal{H}_k(A,B)$
or $\mathcal{I}_k(A,B)$ respectively and that the classic $k$-consistency and
$k$-Weisfeiler-Leman algorithms can be reformulated as computing greatest fixed points
of presheaf operations which successively remove sections which are obstructed from
being part of some global section.  In this section,
we extend these algorithms by considering further efficiently computable obstructions
which arise naturally from presheaf cohomology. From this we derive new cohomological algorithms for CSP and SI.

\subsection{Cohomology and local vs.\ global problems}

The notion of computing cohomology valued in an $\cat{AbGp}$-valued presheaf $\sheaf{F}$
on a topological space $X$ has a long history in algebraic geometry and algebraic topology which
dates back to Grothendieck's seminal paper on the topic~\cite{grothendieck}. The cohomology
valued in $\sheaf{F}$ consists of a sequence of abelian groups $H^i(X,\sheaf{F})$ where
$H^0(X,\sheaf{F})$ is the free $\mathbb{Z}$-module over global sections
of $\sheaf{F}$. As seen in the previous section we may be interested in such global
sections but their existence may be difficult to determine. This is where the
functorial nature of cohomology is extremely
useful. Indeed, any short exact sequence of presheaves
$$ 0 \rightarrow \sheaf{F}_L \rightarrow \sheaf{F} \rightarrow \sheaf{F}_R \rightarrow 0 $$
lifts to a long exact sequence of cohomology groups
$$ 0 \rightarrow H^0(X, \sheaf{F}_L) \rightarrow H^0(X, \sheaf{F}) \rightarrow H^0(X, \sheaf{F}_R) \rightarrow H^1(X, \sheaf{F}_L) \rightarrow \ldots .$$
This tells us that the global sections of $\sheaf{F}_R$ which are not images of
global sections of $\sheaf{F}$ are mapped to non-trivial elements of the group $H^1(X, \sheaf{F}_L)$ by the maps in this sequence. This means that these higher cohomology groups can be seen as a source of
obstacles to lifting ``local'' solutions in $\sheaf{F}_R$ to ``global'' solutions in $\sheaf{F}$.

An important recent example of such an application of cohomology to finite structures
can be found in the work of Abramsky, Barbosa, Kishida, Lal and Mansfield \cite{abramsky_et_al} in quantum foundations. They show
that cohomological obstructions of the type described above can be used to detect
contextuality (locally consistent measurements which are globally inconsistent)
in quantum systems which were earlier given a presheaf semantics
by Abramsky and Brandenburger~\cite{ab_brand}. In Appendix \ref{sec:quant_cont},
we describe these obstructions in general and show how the presheaves we
constructed in the last section admit the same cohomological obstructions. This
similarity inspires the definitions and algorithms which follow in the next two sections.

\subsection{$\mathbb{Z}$-local sections and $\mathbb{Z}$-extendability}

Returning to presheaves of local homomorphisms and isomorphisms let $\sheaf{S}$
be a subpresheaf of $\sheaf{H}_k$. Then we define the presheaf of $\mathbb{Z}$-linear
local sections of $\sheaf{S}$ to be the presheaf of formal $\mathbb{Z}$-linear
sums of local sections of $\sheaf{S}$. This means that for any $C \in A^{\leq k}$
$$\mathbb{Z}\sheaf{S}(C) \defeq \left\{ \sum_{s \in \sheaf{S}(C)} \alpha_s s \ \mid \ \alpha_s \in \mathbb{Z} \right\}.  $$

This is an abelian presheaf on $A^{\leq k}$ and we call the global sections $\{ r_U \in \mathbb{Z}\sheaf{S}(U) \}_{U \in A^{\leq k}}$ $\mathbb{Z}$-linear global sections of $\sheaf{S}$. We say that
a local section $s \in \sheaf{S}(C)$ is $\mathbb{Z}$-extendable if there is a
$\mathbb{Z}$-linear global section $\{ r_U \in \mathbb{Z}\sheaf{S}(U) \}_{U \in A^{\leq k}}$
such that $r_C = s$. We write this condition as $\mathbf{\mathbb{Z}ext}(\sheaf{S}, s)$.
As outlined in Appendix \ref{sec:quant_cont}, this condition corresponds to the
absence of a cohomological obstruction to $\sheaf{S}$ containing a global section
involving $s$.

Importantly for our purposes, deciding the condition $\mathbf{\mathbb{Z}ext}(\sheaf{S}, s)$
for any $\sheaf{S} \subset \sheaf{H}_k(A,B)$ is computable in polynomial time in the
sizes of $A$ and $B$. This is because the compatibility conditions for a collection
$\{ r_U \in \mathbb{Z}\sheaf{S}(U) \}_{U \in A^{\leq k}}$ being a global section of $\mathbb{Z}\sheaf{S}$
can be expressed as a system of polynomially many linear equations in polynomially many variables.
Indeed, we write each $r_U$ as $\sum_{s \in \sheaf{S}(U)} \alpha_s s$ where $\alpha_s$
is a variable for each $s \in \sheaf{S}(U)$. This gives a total number of variables
bounded by $\bigO(\abs{A}^k \cdot \abs{B}^k)$, the size of $\mathbf{Hom}_k(A,B)$.
For each of the $\bigO(\abs{A}^{2k})$ pairs of contexts $U, U' \in A^{\leq k}$,
the compatibility condition $(r_U)_{\mid_{U\cap U'}} = (r_{U'})_{\mid_{U\cap U'}}$
yields a linear equation in the $\alpha_s$ variables for each $s' \in \sheaf{S}(U\cap U')$,
leading to a total number of equations bounded by $\bigO(\abs{A}^{2k}\cdot \abs{B}^k)$.
By an algorithm
of Kannan and Bachem~\cite{KannanBachem} can be solved in polynomial time in the sizes
of $A$ and $B$. This allows us
to define the following efficient algorithms for CSP and SI based on removing cohomological
obstructions.

\subsection{Cohomological algorithms for CSP and SI}

We saw in Section \ref{sec:presheaf} that the $k$-consistency and $k$-Weisfeiler-Leman
algorithms can be recovered as greatest fixpoints of presheaf operators
removing local sections which fail the forth and bijective-forth properties respectively.
Now that we have from cohomological considerations a new necessary condition $\mathbf{\mathbb{Z}ext}(\sheaf{S},s)$
for a local section to feature in a global section of $\sheaf{S}$, we can define
natural extensions to the $k$-consistency and $k$-Weisfeiler-Leman algorithms as
follows.

\subsubsection{Cohomological $k$-consistency}

To define the cohomological $k$-consistency algorithm, we first define an operator
which removes those local sections which admit a cohomological obstruction.
Let $(\cdot)^{\mathbb{Z}\downarrow}$ be the operator which computes for a given presheaf $\sheaf{S} \subset \sheaf{H}_k$ the subpresheaf $\sheaf{S}^{\mathbb{Z}\downarrow}$ where $\sheaf{S}^{\mathbb{Z}\downarrow}(C)$
contains exactly those local sections $s \in \sheaf{S}(C)$ which satisfy both the forth property $\mathbf{Forth}(\sheaf{S}, s)$ and the $\mathbb{Z}$-extendability
property $\mathbf{\mathbb{Z}ext}(\sheaf{S}, s)$. As this process may remove the local sections
in $\sheaf{S}$ which witness the extendability of other local sections we need to
take a fixpoint of this operator to get a presheaf with the right extendability properties
at every local section. So, we write $\overline{\sheaf{S}}^{\mathbb{Z}}$
for the greatest fixpoint of this operator starting from $\sheaf{S}$. As both
$\mathbf{Forth}(\sheaf{S}, s)$ and $\mathbf{\mathbb{Z}ext}(\sheaf{S}, s)$ are both
 computable in polynomial time in the size of $\sheaf{S}$ and $\overline{\sheaf{S}}^{\mathbb{Z}}$
 has a global section if and only if $\sheaf{S}$ has a global section, this allows us to define
 the following efficient algorithm for approximating CSP.

 \begin{defn}
   The \emph{cohomological $k$-consistency algorithm} accepts an instance $(A,B)$
   if the greatest fixpoint $\overline{\sheaf{H}_k(A,B)}^{\mathbb{Z}}$ is non-empty and
   otherwise rejects. \\
   If $(A,B)$ is accepted by this algorithm we write $A \rightarrow^{\mathbb{Z}}_k B$
   and say that the instance $(A, B)$ is cohomologically $k$-consistent.
 \end{defn}

We conclude this section by showing that the relation $\rightarrow^{\mathbb{Z}}_k$ is transitive.

\begin{prop}\label{prop:comp}
  For all $k$, given $A, B$ and $ C$ structures over a common finite signature
  $$ A \rightarrow^{\mathbb{Z}}_k B \rightarrow^{\mathbb{Z}}_k C
  \implies A \rightarrow^{\mathbb{Z}}_k C. $$
\end{prop}

\begin{proof}
See Appendix \ref{sec:app1}.
\end{proof}

\subsubsection{Cohomological $k$-Weisfeiler-Leman}

We now define cohomological $k$-equivalence to generalise $k$-WL-equivalence in the
same way as we did for cohomological $k$-consistency, by removing local sections
which are not $\mathbb{Z}$-extendable. As $\mathbb{Z}$-extendability in $S\subset \mathbf{Isom}_k(A,B)$
is not \textit{a priori} symmetric in $A$ and $B$ we need to check that both
$s$ is $\mathbb{Z}$-extendable in $S$ and $s^{-1}$ is $\mathbb{Z}$-extendable in
$S^{-1} = \{ t^{-1} \ \mid \ t \in S \}$. We call this $s$ being $\emph{$\mathbb{Z}$-bi-extendable}$
in $S$ and write it as $\mathbf{\mathbb{Z}bext}(\sheaf{S}, s)$. We incorporate this into a new presheaf operator $(\cdot)^{\mathbb{Z}\#}$ as follows.
Given a presheaf $\sheaf{S} \subset \sheaf{I}_k$ let $\sheaf{S}^{\mathbb{Z}\#}$ be the largest subpresheaf
of $\sheaf{S}$ such that every $s \in \sheaf{S}^{\mathbb{Z}\#}(C)$
satisfies both the bijective forth property $\mathbf{BijForth}(\sheaf{S}, s)$ and the $\mathbb{Z}$-bi-extendability
property $\mathbf{\mathbb{Z}bext}(\sheaf{S}, s)$. We write $\overline{\overline{\sheaf{S}}}^{\mathbb{Z}}$
for the greatest fixpoint of this operator starting from $\sheaf{S}$. As both
$\mathbf{BijForth}(\sheaf{S}, s)$ and $\mathbf{\mathbb{Z}bext}(\sheaf{S}, s)$ are
 computable in polynomial time in the size of $\sheaf{S}$ and $\overline{\overline{\sheaf{S}}}^{\mathbb{Z}}$
 has a global section if and only if $\sheaf{S}$ has a global section, this allows us to define
 the following efficient algorithm for approximating SI.

 \begin{defn}
   The \emph{cohomological $k$-Weisfeiler-Leman} accepts an instance $(A,B)$
   if the greatest fixpoint $\overline{\overline{\sheaf{I}_k(A,B)}}^{\mathbb{Z}}$ is non-empty and
   otherwise rejects. \\
   If $(A,B)$ is accepted by this algorithm we write $A \equiv^{\mathbb{Z}}_k B$
   and say that the instance $(A, B)$ is cohomologically $k$-equivalent.
 \end{defn}

 Finally, we observe that the existence of a non-empty subpresheaf of $\sheaf{I}_k$
 satisfying the $\mathbf{BijForth}$ and $\mathbf{\mathbb{Z}bext}$ properties also
 satisfies the conditions for witnessing
 cohomological $k$-consistency of the pairs $(A, B)$ and $(B, A)$.
 Formally we have
 \begin{obs}\label{obs:backforth}
  For any two structures $A$ and $B$, $A \equiv^{\mathbb{Z}}_k B$ implies that
  $A \rightarrow^{\mathbb{Z}}_k B$ and $B \rightarrow^{\mathbb{Z}}_k A$.
 \end{obs}

 In Section \ref{sec:effect}, we will demonstrate the power of these new algorithms
 by showing that both cohomological $k$-consistency and cohomological $k$-Weisfeiler-Leman
 can solve instances of CSP and SI on which the non-cohomological versions fail. Before
 doing this, we briefly review some other algorithms for CSP and SI which
 involve solving systems of linear equations and establish a possible connection to
 be explored in future work.

\subsection{Other algorithms for CSP and SI}
While the connections to cohomology in approximating CSP and SI are novel in this
paper, the algorithms introduced here are not the first to use solving systems of
linear equations to approximate these problems.

On the CSP side, some examples of such algorithms include the BLP+AIP~\cite{aip} and CLAP~\cite{clap} algorithms studied in the Promise CSP community. One difference here is that
for an instance $(A,B)$ the variables in BLP and AIP are indexed by valid assignments to each variable and to each related tuple instead of being indexed by valid $k$-local homomorphisms as in the algorithm derived above.
This means that directly comparing these algorithms as stated is not straightforward and is beyond the scope of this paper. However,
it seems likely that these algorithms can also be expressed in terms of appropriate presheaves.
For example, let $\cat{C(A)}$ be the category whose objects are the elements of $A$
and the related tuples of $A$ and with maps for each projection from a related tuple to
an element, and let the $\cat{Set}$-valued presheaf $\sheaf{H}_{C}(A,B)$ on $\cat{C(A)}$
 map any $a \in A$ to the set of all elements in $B$ and any $\mathbf{a}\in R^A$
to the set of all related tuples $R^B$. Then, in a similar way to above, we can see that
global sections of $\sheaf{H}_{C}$ are homomorphisms from $A$ to $B$. In future
work, we will compare the fixpoints $\overline{\sheaf{H}_{C}}$ and $\overline{\sheaf{H}_{C}}^{\mathbb{Z}}$
with solutions to the BLP and AIP systems of equations and we will explore a possible
presheaf representation for CLAP.

On the SI side, Berkholz and Grohe~\cite{BerkholzGrohe} have studied $\mathbb{Z}$-linear
versions of the Sherali-Adams hierarchy of relaxations of the graph isomorphism problem.
They establish that no level of this hierarchy decides the full isomorphism relation on graphs.
Their algorithm for the $k^{\text{\small th}}$ level of the hierarchy appears similar to
checking the $\mathbb{Z}$-extendability in $\sheaf{H}^{A,B}_k$ of the empty solution $\epsilon \in \sheaf{H}^{A,B}_k(\emptyset)$. A full comparison of this algorithm and the algorithm
described above is an interesting direction for future work.

\section{The (unreasonable) effectiveness of cohomology in CSP and SI}\label{sec:effect}


In this section, we prove that the new algorithms arising from this cohomological approach to CSP and SI are substantially more powerful than the $k$-consistency and $k$-Weisfeiler-Leman algorithms. In particular, we show that cohomological $k$-consistency resolves CSP over all domains of arity less than or equal to $k$ which admit a ring representation and that for a fixed small $k$ cohomological $k$-Weisfeiler-Leman can distinguish structures which differ on a very general form of the CFI property, in particular, showing that cohomological $k$-Weisfeiler-Leman can distinguish a property which Lichter~\cite{lich} claims not to be expressible in rank logic.

\subsection{Cohomological $k$-consistency solves all affine CSPs}

In this section, we demonstrate the power of the cohomological $k$-consistency
algorithm by proving that it can decide the solvability of systems of equations
over finite rings.

To express the main theorem of this section in terms of the finite relational structures
on which our algorithm is defined, we first need to fix a notion of ring representation
of a relational structure. Let $A$ be a relational structure over signature $\sigma$
with relations given by $\{R^A \}_{R\in \sigma}$. We say that $A$ has a \emph{ring representation}
if we can give the set $A$ a ring structure $(A, + ,\cdot, 0, 1)$ such that
for every relational symbol $R\in \sigma$ the set $R^A \subset A^m$ is an affine
subset of the ring $(A, +, \cdot, 0 , 1)$, meaning that there exists $b^R_1, \ldots
, b^R_m, a^R \in A$ such that

$$ R^A = \big\{ \mathbf{x} \in A^m \ \mid \ \sum_{i\in [m]} b^R_i \cdot x_i = a^R   \big\} $$

With this necessary background we state the main theorem of this section.

\begin{theor} \label{thm:rings}
  For any structure $B$ with a ring representation,
  there is a $k$ such that the cohomological $k$-consistency algorithm decides $\mathbf{CSP}(B)$.\\
  Alternatively stated, there exists a $k$ such that for all $\sigma$-structures $A$
  $$ A \rightarrow^{\mathbb{Z}}_{k} B \iff A \rightarrow B $$
\end{theor}

\begin{proof}
  See Appendix \ref{sec:app2}.
\end{proof}

This theorem is notable because there are relational structures $B$ with ring representations
for which there are families of structures $A_k$ such that $A_k \rightarrow_k B$ but $A_k \not \rightarrow B$, see for example the examples given by Feder and Vardi~\cite{federvardi}.
Furthermore, there exist pairs $(A_k, B_k)$ where $A_k \equiv_k B_k$, $B_k \rightarrow B$ and
$A_k \rightarrow_k B$ but $A_k \not \rightarrow B$, see for example the
work of Atserias, Bulatov and Dawar~\cite{AtseriasBulatovDawar}.
As the sequence of relations $\equiv_k$ bounds
the expressive power of FPC, this effectively proves that
the solvability of systems of linear equations over $\mathbb{Z}$, which is central
to the cohomological $k$-consistency algorithm, is not expressible in FPC. This result
 was not previously known to the author.

\subsection{Cohomological $k$-Weisfeiler-Leman decides the CFI property}

The Cai-Fürer-Immerman construction~\cite{cfi} on ordered finite graphs is a very powerful
tool for proving expressiveness lower bounds in descriptive complexity theory. While it
was originally used to separate the infinitary $k$ variable logic with counting from
\texttt{PTIME}, it has since been used in adapted forms to prove bounds on invertible maps
equivalence~\cite{dgp}, computation on Turing machines with atoms~\cite{bojan} and rank logic~\cite{lich}.
In this section, we show that $\equiv^{\mathbb{Z}}_k$ separates a very general form
of this

The version we consider in this paper is parameterised by a prime power $q$ and
takes any totally ordered graph $(G, <)$ and any map $g \colon E(G) \rightarrow \mathbb{Z}_q$
to a relational structure $\mathbf{CFI}_q(G,g)$. The construction effectively encodes
a system of linear equations over $\mathbb{Z}_q$ based on the edges of $G$ and
the ``twists'' introduced by the labels $g$. The result is the following well-known
fact.

\begin{fact}\label{fact:cfi}
  For any prime power, $q$, ordered graph $G$, and functions $g, h \colon E(G) \rightarrow  \mathbb{Z}_q$,
  $$ \mathbf{CFI}_{q}(G, g) \cong \mathbf{CFI}_{q}(G, h) \iff \sum g = \sum h $$
\end{fact}

We say that the structure $\mathbf{CFI}_{q}(G, g)$ has the \emph{CFI property}
if $\sum g = 0$. For more details on this construction we refer to Appendix \ref{sec:app3}
or the recent paper of Lichter~\cite{lich} whose presentation we follow.

We now recall the two major separation results based on this construction. The
first is a landmark result of descriptive complexity from the early 1990's.

\begin{theor}[Cai, F\"urer, Immerman~\cite{cfi}]
  There is a class of ordered (3-regular) graphs $\mathcal{G} = \{ G_n \}_{n\in \mathbb{N}}$
  such that in the respective class of CFI structures
  $$\mathcal{K} = \{ \mathbf{CFI}_2(G,g) \ \mid \ G \in \mathcal{G}, g\colon V(G) \rightarrow \mathbb{Z}_2   \} $$
  the CFI property is decidable in polynomial-time but cannot be
  expressed in FPC.
\end{theor}
The second is a recent breakthrough due to Moritz Lichter.
\begin{theor}[Lichter~\cite{lich}]\label{thm:lich}
There is a class of ordered graphs $\mathcal{G} = \{ G_n \}_{n\in \mathbb{N}}$
such that in the respective class of CFI structures
$$\mathcal{K} = \{ \mathbf{CFI}_{2^k}(G,g) \ \mid \ G \in \mathcal{G}  \} $$
the CFI property is decidable in polynomial-time (indeed, expressible in
choiceless polynomial time) but cannot be expressed in rank logic.
\end{theor}

Despite this CFI property proving to be inexpressible in both FPC and rank
logic, we show that (perhaps surprisingly) there is a fixed $k$ such that
cohomological $k$-Weisfeiler-Leman algorithm can separate structures which differ on
this property in the following general way. The proof of this theorem relies the on showing
that $\equiv^{\mathbb{Z}}_k$ behaves well with logical interpretations and the details are left
to Appendix \ref{sec:app3}.

\begin{theor}\label{thm:main}
  There is a fixed $k$ such that for any $q$ given $\mathbf{CFI}_{q}(G,g)$ and
  $\mathbf{CFI}_{q}(G,h)$ with $\sum g = 0 $ we have
  $$ \mathbf{CFI}_{q}(G,g) \equiv^{\mathbb{Z}}_k \mathbf{CFI}_{q}(G,h) \iff \mathbf{CFI}_{q}(G,g) \cong \mathbf{CFI}_{q}(G,h)$$
\end{theor}

\begin{proof}
  See Appendix \ref{sec:app3}.
\end{proof}

As a direct consequence of this result, there is some $k$ such that the set of structures with the CFI property
in Lichter's class $\mathcal{K}$ from Theorem \ref{thm:lich} is closed under $\equiv^{\mathbb{Z}}_k$.
This means that, by the conclusion of Theorem \ref{thm:lich}, the equivalence relation $\equiv^{\mathbb{Z}}_k$
can distinguish structures which disagree on a property that is not expressible in rank logic. Indeed, Dawar, Gr\"adel and Lichter~\cite{dgl} show further that this property is also inexpressible in linear algebraic logic.
By the definition of our algorithm for $\equiv^{\mathbb{Z}}_k$ this implies that
solvability of systems of $\mathbb{Z}$-linear equations is not definable in linear algebraic logic.

\section{Conclusions \& future work}\label{sec:conclusion}


In this paper, we have presented novel approach to CSP and SI in terms of presheaves and
have used this to derive efficient generalisations of the $k$-consistency
and $k$-Weisfeiler-Leman algorithms, based on natural considerations of presheaf cohomology. We have shown
that the relations, $\rightarrow^{\mathbb{Z}}_k$ and $\equiv^{\mathbb{Z}}_k$, computed
by these new algorithms are strict refinements of their well-studied classical counterparts
$\rightarrow_k$ and $\equiv_k$. In particular, we have shown in Theorem \ref{thm:rings} that
cohomological $k$-consistency suffices to solve linear equations over all finite rings
and in Theorem \ref{thm:main} that cohomological $k$-Weisfeiler-Leman distinguishes positive and
negative instances of the
CFI property on the classes of structures studied by Cai, F\"urer and Immerman~\cite{cfi}
and more recently by Lichter~\cite{lich}. These results have important consequences for
descriptive complexity theory showing, in particular, that the solvability of
systems of linear equations over $\mathbb{Z}$ is not expressible in FPC, rank logic
or linear algebraic logic. Furthermore, the results
of this paper demonstrate the unexpected effectiveness of a cohomological approach to
constraint satisfaction and structure isomorphism, analogous to that pioneered by
Abramsky and others for the study of quantum contextuality. \\

The results of this paper suggest several directions for future work to establish the
extent and limits of this cohomological approach. We ask the following questions
which connect it to important themes in algorithms, logic and finite model theory.

\textbf{Cohomology and constraint satisfaction}:
Firstly, Bulatov and Zhuk's recent independent resolutions of the Feder-Vardi conjecture~\cite{Bulatov,Zhuk},
show that for all domains $B$ either $\mathbf{CSP}(B)$ is \texttt{NP-Complete}
or $B$ admits a weak near-unanimity polymorphism and $\mathbf{CSP}(B)$ is tractable.
As the cohomological $k$-consistency algorithm expands the power of the
$k$-consistency algorithm which features as one case of Bulatov and Zhuk's general
efficient algorithms, we ask if it is sufficient to decide all
tractable $\mathbf{CSP}$s.

\begin{quest}
  For all domains $B$ which admit a weak near-unanimity polymorphism, does there exists
  a $k$ such that for all $A$ $$A \rightarrow B \iff A \rightarrow^{\mathbb{Z}}_k B ? $$
\end{quest}

\textbf{Cohomology and structure isomorphism}:
Secondly, as cohomological $k$-Weisfeiler-Leman is an efficient algorithm
for distinguishing some non-isomorphic relational structures we ask if it distinguishes
all non-isomorphic structures. As the best known structure isomorphism algorithm is
quasi-polynomial~\cite{babai}, we do not expect a positive answer to this question but expect that
negative answers would aid our understanding of the hard cases of structure isomorphism in general.
\begin{quest}
  For every signature $\sigma$ does there exists a $k$ such that for all $\sigma$-structures $A, B$
  $$ A \cong B \iff A \equiv^{\mathbb{Z}}_k B? $$
\end{quest}

\textbf{Cohomology and game comonads}:
Thirdly, as $\rightarrow_k$ and $\equiv_k$ have been shown by Abramsky, Dawar, and Wang~\cite{pebb}
to be correspond to the coKleisli morphisms and isomorphisms of a comonad $\mathbb{P}_k$,
we ask whether a similar account can be given to $\rightarrow^{\mathbb{Z}}_k$ and
$\equiv^{\mathbb{Z}}_k$. As the coalgebras of the $\mathbb{P}_k$ comonad relate to the
combinatorial notion of treewidth, an answer to this question could provide a new
notion of ``cohomological'' treewidth.

\begin{quest}
  Does there exist a comonad $\mathbb{C}_k$ for which the notion of morphism and
  isomorphism in the coKleisli category are $\rightarrow^{\mathbb{Z}}_k$ and
  $\equiv^{\mathbb{Z}}_k$?
\end{quest}

\textbf{The search for a logic for \texttt{PTIME}}:
Finally, as the algorithms for $\rightarrow^{\mathbb{Z}}_k$ and
$\equiv^{\mathbb{Z}}_k$ are likely expressible in rank logic extended with a quantifier for
solving systems of linear equations over $\mathbb{Z}$ and as $\equiv^{\mathbb{Z}}_k$
distinguishes all the best known family separating rank logic from \texttt{PTIME},
we ask if solving systems of equations over $\mathbb{Z}
$ is enough to
capture all \texttt{PTIME} queries.
\begin{quest}
  Is there a logic FPC+rk+$\mathbb{Z}$ incorporating solvability of $\mathbb{Z}$-linear equations
  into rank logic which captures \texttt{PTIME}?
\end{quest}

\bibliography{refs}

\newpage

\appendix

\section{Proof omitted from Section \ref{sec:presheaf} }\label{sec:app0}

\begin{proof}[Proof of Lemma \ref{lem:klocal}]
($\implies$) This direction is easy. In the case of homomorphisms the argument
proceeds as follows. Suppose that $(A, B) \in CSP$ and so there exists $h : A \rightarrow B$ a homomorphism. Consider the collection of maps $\{ h_U \}_{U \in A^{\leq k}}$ defined by $h_U = h_{\mid_{U}}$. This forms global section of $\sheaf{H}_k$ because, firstly, $h_U \in \sheaf{H}_k(U)$ as the restriction of a homomorphism is a partial homomorphism and, secondly, the naturality condition is satisfied as $(h_U)_{\mid_{U'}} = h_{\mid_{U'}}$ for any $U' \subset U$.
The argument follows similarly for $SI$ and $\sheaf{I}_k$.\\
($\impliedby$) For this direction, in the case of homomorphisms, we have global section $s : \mathbb{I} \implies \sheaf{H}_k$. Firstly, we claim that there is a single function $h: A \rightarrow B$ such that $s_U = h_{\mid_{U}}$ for all $U \in A^{\leq k}$. Indeed, this is the function $h$ which sends any element $a \in A$ to the element $h(a)  \defeq s_{ \{ a \} }(a) \in B$.
This satisfies the required property that $s_U = h_{\mid_{U}}$, as for any $U \in A^{\leq k}$ and any $u \in U$, naturality of $s$ along the inclusion $\{u\}  \subset U$ ensures that $s_U(u) = s_{\{u\}}(u) = h(u)$ and so $s_U = h_{\mid_{U}}$.
To show that $h$ is a homomorphism, take any related tuple $(a_1, \ldots, a_m) \in R^A$.
Let $U = \{a_1, \ldots a_m \}$. As $k$ is at least the arity of $\sigma$ we have that $\abs{U} \leq m \leq k$ and so $U \in A^{\leq k}$. Now, $h_{\mid_U} = s_U \in \mathbf{Hom}_k(A,B)$ is a partial homomorphism. So, $(a_1, \ldots, a_m) \in R^A \implies (b_1, \ldots, b_m) \in R^B$. Thus $h$ is a homomorphism.

In the case of isomorphisms, this we can define the map $h$ in the same way. As its
projections $s_U = h_{\mid_{U}}$ are all partial isomorphisms, $h$ will be injective
and so is bijective by the assumption on sizes of $A$ and $B$. So applying the above,
$h$ will be a bijective homomorphism. To show that it is indeed an isomorphism
we show that it reflects related tuples in $B$. Suppose $(b_1, \ldots, b_m) \in R^B$
then consider the set  $V = \{b_1, \ldots, b_m \}$. As $h$ is bijective there is
a set $U = h^{-1}(V) \in A^{\leq k}$. Then $h_{\mid_{U}} = s_U$ is an isomorphism
between $U$ and $V$ and so $(h^{-1}(b_1), \ldots, h^{-1}(b_m)) \in R^A$ as required.
\end{proof}

\section{Algorithms for $k$-consistency and $k$-Weisfeiler-Leman}\label{sec:app_algs}

In this appendix, we recall the full definitions of $k$-consistency and
$k$-Weisfeiler-Leman.

\subsection{Classical $k$-consistency algorithm}
We start by recalling some definitions related to the
classical $k$-consistency algorithm on which our algorithm will build.

For $A$ and $B$ finite structures over a common (finite) signature,
let $\mathbf{Hom}_k(A, B)$ denote the set of partial homomorphisms
from $A$ to $B$ with domain of size less than or equal to $k$. There
is a natural partial order $<$ on this set, defined
as follows.  For any partial homomorphisms $f,g \in \mathbf{Hom}_k(A, B)$
we say that $f < g$ if $\mathbf{dom}(f) \subset \mathbf{dom}(g)$
and $g_{\mid_{\mathbf{dom}(f)}} = f$.

We say that any $S \subset \mathbf{Hom}_k(A,B)$ has the \emph{forth}
property if for every $f \in S$ with $|\mathbf{dom}(f)|< k$
we have the property $\mathbf{Forth}(S,f)$ which is defined as follows:
$$  \ \forall a \in A, \
        \exists b \in B \text{ s.t. } f \cup \{(a,b)\} \in S. $$

Given $S \subset \mathbf{Hom}_k(A,B)$ we define $\overline{S}$ to
be the largest subset of $S$ which is downwards-closed and has the
forth property. Note that $\emptyset$ satisfies these conditions,
so such a set always exists. For a fixed $k$ there is a simple
algorithm for computing $\overline{S}$ from $S$.

This is done by starting with $S_0 = S$ and then entering the
following loop with $i=0$
\begin{enumerate}
  \item Initialise $S_{i+1}$ as being equal to $S_i$.
  \item For each $s \in S_{i}$, check if $\mathbf{Forth}(S_i,s)$ holds and
  if not remove it from $S_{i+1}$ along with all $s'> s$.
  \item If none fail this test, halt and output $S_i$.
  \item Otherwise, increment $i$ by one and repeat.
\end{enumerate}
It is easily seen that this runs in polynomial time in $\abs{A}\abs{B}$.

Now for a pair of structures $A, B$ we say that the pair
$(A,B)$ is $k$ consistent if $\overline{\mathbf{Hom}_k(A,B)} \neq \emptyset$.
We denote this by writing $A \rightarrow_k B$ and the algorithm above
shows how to decide this relation in polynomial time for fixed $k$.
This relation has many equivalent logical and algorithmic definitions
as seen in \cite{federvardi}, and \cite{BartoKozik}.

\subsection{Classical $k$-Weisfeiler-Leman algorithm}

Immerman and Lander\cite{IL} first established that two structures are
$\equiv_{k-WL}$-equivalent if and only if they satisfy the same formulas of infinitary
$k$-variable logic with counting quantifiers (written $A \equiv_{k} B$).
Hella\cite{lhipt} showed that this is true if and only if the set of $k$-local partial
isomorphisms $\mathbf{Isom}_{k}(A,B)$ contains a non-empty subset $S$ which is
downward-closed and has the following \emph{bijective forth property} for all
$f \in S$ with $|\mathbf{dom}(f)|< k$:
$$ \ \exists b_f: A \rightarrow B \text{ a bijection s.t.\ }
\forall a \in A \ f \cup \{(a,b_f(a))\} \in S $$
Whether such a bijection exists can be determined efficiently given $A, B, S $ and
$f$ by determining if the bipartite graph with vertices $A \sqcup B$ and edges
$\{(a,b) \ \mid \ f \cup \{(a,b)\} \in S \}$ has a perfect matching. For $S \subset \mathbf{Isom}_k(A,B)$,
let $\overline{\overline{S}}$ be the largest subset of $S$ which is downward-closed
and satisfies the bijective forth property. For fixed $k$ this can be computed
in polynomial time in the sizes of $A$ and $B$ and so an alternative polynomial time
algorithm for determining $\equiv_{k-WL}$ is computing $\overline{\overline{\textbf{Isom}_{k}(A,B)}}$
and checking if it is non-empty.\\

\section{Cohomological obstructions from quantum contextuality}\label{sec:quant_cont}

To understand the cohomological invariants of Abramsky, Barbosa and Mansfield\cite{ab_man_barb} which we need for the main algorithms in this section we first give a brief overview the sheaf-theoretic approach to
quantum contextuality introduced by Abramsky and Brandenburger\cite{ab_brand}
which bears an important resemblance to the set-up in the last section. \\

A \emph{measurement scenario} is a triple $\mathcal{M} = \langle X, M, O  \rangle$ where $X$ and $O$ are finite sets and $M$ is a downward-closed subset of the powerset $P(X)$ which covers $X$. We interpret such a scenario as a quantum system with a set $X$ of possible measurements, a set $M$ of valid contexts of measurements that can be performed simultaneously, and a set of outcomes $O$ for each measurement. The \emph{sheaf of outcomes} over $\mathcal{M}$ is the presheaf $\sheaf{E} : \cat{M}^{op} \rightarrow \cat{Set}$ defined by $\sheaf{E}(C) = O^C$ with the restriction maps given
by normal function restriction. The proof that this is indeed a sheaf is elementary but unimportant for the present work. A \emph{possibilistic empirical model} of $\mathcal{M}$ is any flasque subpresheaf $\sheaf{S}$ of $\sheaf{E}$. For any such model we interpret the set of local sections $\sheaf{S}(C) \subset O^C$ as the set of possible measurement-outcome pairs for the context $C$. The condition of being flasque is precisely what's required for such a model to satisfy the no-signalling property which is important in quantum mechanical systems.
As in the previous section, global sections of these presheaves are important. Indeed Abramsky, Barbosa and Mansfield say that an empirical model $\sheaf{S}$ is \emph{strongly contextual}, written $\mathbf{SC}(\sheaf{S})$ if there is no global section $\{s_C \in \sheaf{S}(C)\}_{C\in M}$ for $\sheaf{S}$. Furthermore, a possible measurement outcome $s \in \sheaf{S}(C')$ is said to be \emph{logically contextual}, written $\mathbf{LC}(\sheaf{S}, s)$ if there is no global section $\{s_C \in \sheaf{S}(C)\}_{C\in M}$ for $\sheaf{S}$ such that $s_{C'} = s$. The whole empirical model $\sheaf{S}$ is said to be \emph{logically contextual}, written $\mathbf{LC}(\sheaf{S})$ if there exists some local section $s$ of $\sheaf{S}$ such that $LC(\sheaf{S}, s)$ holds.\\

In \textit{The Cohomology of Non-locality and Contextuality}, Abramsky, Barbosa and Mansfield show that contextuality in empirical models, as defined above, can be detected in many cases by considering the cohomology of certain Čech
cochain complexes $\check{C}^{\bullet}(M, \sheaf{F})$ of the cover $M$ valued in abelian presheaves related to $\sheaf{S}$. To do this they first define, for any possibilistic empirical model $\sheaf{S}$, the abelian presheaf $\sheaf{F}_{\mathbb{Z}} : \cat{M}^{op} \rightarrow \cat{AbGrp} $  which is formed by composing $\sheaf{S}$ with the free $\mathbb{Z}$-module functor $F_{\mathbb{Z}} : \cat{Set} \rightarrow \cat{AbGrp}$. Local sections $r \in \sheaf{F}_Z(C)$ are simply formal $\mathbb{Z}$-linear combinations of elements of $\sheaf{S}(C)$.
For any $U \in M$, they then construct a short exact sequence
$$ 0 \rightarrow \sheaf{F}_{\tilde{U}} \rightarrow \sheaf{F}_{\mathbb{Z}} \rightarrow \sheaf{F}_{\mid_U} \rightarrow 0 $$
which captures the restriction of local sections to the context $U$.
 This gives a long exact sequence of cohomology groups. The connection maps in this long exact sequnce allows us to take any $s \in \sheaf{S}(U)$ and send it forward to an element $\delta(s) \in \check{H}^1(M, \sheaf{F}_{\tilde{U}})$. Abramsky et al show that $\delta(s)$ not vanishing
is a sufficient condition for $\mathbf{LC}(\sheaf{S}, s)$ and define this condition as $\textbf{CLC}_{\mathbb{Z}}(\sheaf{S}, s)$.
They also give the following equivalent condition which we use for the rest of the paper.
$\textbf{CLC}_{\mathbb{Z}}(\sheaf{S}, s)$ holds if and only if there is no global section
$\{ r_C \}_{C\in M}$ of $\sheaf{F}_{\mathbb{Z}}$ such that $r_U = s$.

Now we see how this set-up applies equally to the search for global sections
in CSP and SI.

\subsection{$\mathbb{Z}$-extendability and $\mathbb{Z}$-linear sections}

In order to translate the cohomological obstructions from the setting of quantum contextuality
to that of constraint satisfaction and structure isomorphism, we first make the
following observation.

\begin{obs}
  For any two relational structures $A$ and $B$ and any $k$, the sheaf of events $\sheaf{E}_{\mathcal{M}}$
  over the measurement scenario $\mathcal{M} = \langle A, A^{\leq k}, B \rangle$
  contains both $\sheaf{H}_k(A,B)$ and $\sheaf{I}_k(A,B)$ as subpresheaves.

  Furthermore, as the subpresheaves $\overline{\sheaf{H}_k}$ and $\overline{\overline{\sheaf{I}_k}}$
  resulting from the sheaf-theoretic versions of $k$-consistency and $k$-Weisfeiler-Leman
  are flasque, they can be viewed as empirical models for $\mathcal{M}$.
\end{obs}

This observation combined with Lemma \ref{lem:klocal} shows that for $k$
at least as large as the arity of the signature of $A$ and $B$, strong contextuality
of the empirical models $\overline{\sheaf{H}_k}$ and $\overline{\overline{\sheaf{I}_k}}$ is equivalent to the pair $(A, B)$ being rejected by CSP and SI, respectively.
Formally this is stated as

\begin{obs}
  For any $A$ and $B$ relational structures and $k$ at least the arity of the largest relation on
  $A$ then
  $$ \mathbf{SC}(\overline{\sheaf{H}_k}(A,B)) \iff A \not\rightarrow B $$
  and
  $$ \mathbf{SC}(\overline{\sheaf{I}_k}(A,B)) \iff A \not\cong B $$
\end{obs}

Furthermore, the logical contextuality of an individual local section
corresponds to the impossibility of extending that section to a full isomorphism
or homomorphism.

\begin{obs}
For any $A$ and $B$ relational structures, $s\in \overline{\sheaf{H}_k}(A,B)(C)$ and
$s'\in \overline{\sheaf{I}_k}(A,B)(C)$ then
$$ \mathbf{LC}(\overline{\sheaf{H}_k}(A,B), s) \iff \neg\exists f : A \rightarrow B \text{ s.t.\ } f_{\mid_C} = s$$
and
$$ \mathbf{LC}(\overline{\sheaf{I}_k}(A,B), s') \iff \neg\exists f : A \rightarrow B \text{, an isomorphism s.t.\ } f_{\mid_C} = s' $$
\end{obs}

As cohomological contextuality gives a sufficient condition for logical
contextuality, we now introduce some terminology for cohomological contextuality
in subpresheaves $\sheaf{S} \subset \sheaf{H}_k(A,B)$. Firstly,
for the abelian presheaf $\sheaf{F} = F_{\mathbb{Z}}\circ \sheaf{S}$,
we call any element $r_C \in \sheaf{F}(C) $ a \emph{$\mathbb{Z}$-linear section}
of $\sheaf{S}$. Such a $\mathbb{Z}$-linear section can be represented as
a formal linear sum $$r_C = \sum_{s\in \sheaf{S}(C)}\alpha_s s$$
where $\alpha_s \in \mathbb{Z}$ for each $s \in \sheaf{S}(C)$. We
say that some $s \in \sheaf{S}(C)$ is $\mathbb{Z}$-extendable
in $\sheaf{S}$, write $\mathbf{\mathbb{Z}ext}(\sheaf{S}, s)$ if there is a collection $\{ r_{C'} \in \sheaf{F}(C') \}_{C' \in M} $
such that $r_C = s$ and for all $C', C'' \in M$ we have
$$ (r_{C'})_{\mid_{C' \cap C''}} =  (r_{C''})_{\mid_{C' \cap C''}}.$$
The following observation is immediate from this definition

\begin{obs}
  For any flasque subpresheaf $\sheaf{S}\subset \sheaf{H}_k(A,B)$ and
  any $s \in \sheaf{S}(C)$, we have
  $$ \mathbf{\mathbb{Z}ext}(\sheaf{S}, s) \iff \neg \mathbf{CLC}_{\mathbb{Z}}(\sheaf{S}, s) $$
\end{obs}

This motivates the definitions of the cohomological algorithms given in the
main paper.

\section{Proofs omitted from Section \ref{sec:cohom} }\label{sec:app1}

To aid with the proof of this proposition we observe that the $\mathbb{Z}$-extendability
condition subsumes both the forth property and downward closure meaning that we have
a slightly simpler condition for the success of the cohomological $k$-consistency algorithm
given as follows.

\begin{obs}\label{obs:consist_set}
  For any structures $A$ and $B$ $A \rightarrow^{\mathbb{Z}}_k B$
  if and only if there exists a set $\emptyset \neq S \subset \mathbf{Hom}_k(A,B)$
 in which each element $s \in S$ is $\mathbb{Z}$-extendable in $S$.
\end{obs}

\begin{proof}[Proof of Proposition \ref{prop:comp}]
Success of the $\rightarrow^{\mathbb{Z}}_k$ algorithm for the pairs $(A,B)$
and $(B, C)$ results in two non-empty sets $S^{AB} \subset \mathbf{Hom}_k(A,B)$
and $S^{BC} \subset \mathbf{Hom}_k(B, C)$ in both of which each local section is
$\mathbb{Z}$-extendable. By Observation \ref{obs:consist_set}, to show that $A \rightarrow^{\mathbb{Z}}_k C$,
it suffices to show that the set $S^{AC} =
\{ s \circ t \ \mid \ s \in S^{BC}, t \in S^{AB} \}$ has the same property. \\

To show that every $p_0 =  s_0 \circ t_0 \in S^{AC}_{\mathbf{a}_0}$ is $\mathbb{Z}$-extendable in $S^{AC}$ we
construct a global $\mathbb{Z}$-linear section extending $p_0$ from the
$\mathbb{Z}$-linear sections $\{r^{t_0}_{\mathbf{a}} := \sum_t z_t t \}_{\mathbf{a}\in A^{\leq k}}$
and $\{r^{s_0}_{\mathbf{b}} := \sum_s w_s s \}_{\mathbf{b}\in B^{\leq k}}$ extending $t_0$ and $s_0$ respectively.
Define $\{ r^{p_0}_{\mathbf{a}} \}_{\mathbf{a} \in A^{\leq k}}$ as
$$ r^{p_0}_{\mathbf{a}} = \sum_{t \in S^{AB}_\mathbf{a}} \sum_{s \in S^{BC}_{t(\mathbf{a})}} z_t w_s (s \circ t)  $$
To show that this is a global $\mathbb{Z}$-linear section extending
$p_0$ we need to show firstly that $r^{p_0}_{\mathbf{a}_0} = p_0$
and secondly that the local sections of $r^{p_0}$ agree on the pairwise
intersections of their domains.

To show that $r^{p_0}_{\mathbf{a}_0} = p_0$ we observe that, as
$r^{t_0}$ $\mathbb{Z}$-linearly extends $t_0$, for all $t \in S^{AB}_{\mathbf{a}_0}$ we have
$$ z_t = \begin{cases}
      1, & \text{for } t = t_0\\
      0, & \text{otherwise,}
        \end{cases} $$
and similarly, for all $s \in S^{BC}_{t_0(\mathbf{a}_0)}$
$$ w_s = \begin{cases}
      1, & \text{for } w = w_0\\
      0, & \text{otherwise.}
        \end{cases} $$
From this we have that
$$ r^{p_0}_{\mathbf{a}_0} = z_{t_0} w_{s_0} (s_0 \circ t_0) = p_0 $$
as required.

Finally, we need to show for any $\mathbf{a}, \mathbf{a}^{\prime}$ in $A^{\leq k}$
with intersection $\mathbf{a}^{\prime\prime}$ that
$$r^{p_0}_{\mathbf{a}\mid_{\mathbf{a}^{\prime\prime}}} = r^{p_0}_{\mathbf{a}^{\prime}\mid_{\mathbf{a}^{\prime\prime}}}.$$
To do this we show that the left hand side depends only
on $\mathbf{a}^{\prime\prime}$ and not on $\mathbf{a}$. As this argument
applies equally to the right hand side, the result follows.

To begin with the left hand side is a dependent sum which loops over $t \in S^{AB}_{\mathbf{a}}$ and
$s \in S^{BC}_{t(\mathbf{a})}$  as follows:
$$
  r^{p_0}_{\mathbf{a}\mid_{\mathbf{a}^{\prime\prime}}} = \sum_{t, s} w_s z_t (s\circ t)_{\mid_{\mathbf{a}^{\prime\prime}}}
$$
To emphasise the dependence on $\mathbf{a}^{\prime\prime}$ we
can group this sum together by pairs $t^{\prime\prime}, s^{\prime\prime}$
with $t^{\prime\prime} \in S^{AB}_{\mathbf{a}^{\prime\prime}}$
and $s^{\prime\prime} \in S^{BC}_{t^{\prime\prime}(\mathbf{a}^{\prime\prime})}$.
Within each group the the sum loops over $t \in S^{AB}_{\mathbf{a}}$ such that
$t_{\mid_{\mathbf{a}^{\prime\prime}}} = t^{\prime\prime}$ and
$s \in S^{AB}_{t(\mathbf{a})}$ such that $s_{\mid_{\mathbf{a}^{\prime\prime}}}
= s^{\prime\prime}$. We write this as

$$
  \sum_{t^{\prime\prime}, s^{\prime\prime}} \sum_{t_{\mid_{\mathbf{a}^{\prime\prime}}} = t^{\prime\prime}} z_t \sum_{s_{\mid_{t^{\prime\prime}(\mathbf{a}^{\prime\prime})}}= s^{\prime\prime}} w_s  (s\circ t)_{\mid_{\mathbf{a}^{\prime\prime}}}
$$
We now show that for each $t^{\prime\prime}, s^{\prime\prime}$ the
corresponding part of the sum depends only on $t^{\prime\prime}$ and $ s^{\prime\prime}$.
This follows from three observations.

The first observation is that in the sum
$$
    \sum_{t_{\mid_{\mathbf{a}^{\prime\prime}}} = t^{\prime\prime}} z_t \sum_{s_{\mid_{t^{\prime\prime}(\mathbf{a}^{\prime\prime})}}= s^{\prime\prime}} w_s  (s\circ t)_{\mid_{\mathbf{a}^{\prime\prime}}}
$$
the formal variables $(s\circ t)_{\mid_{\mathbf{a}^{\prime\prime}}}$
are, by definition, all equal to the variable $(s^{\prime\prime}\circ t^{\prime\prime})$.
Thus we need only consider the coefficients, given by the sum
$$ \sum_{t_{\mid_{\mathbf{a}^{\prime\prime}}} = t^{\prime\prime}} z_t \sum_{s_{\mid_{t^{\prime\prime}(\mathbf{a}^{\prime\prime})}}= s^{\prime\prime}} w_s $$

The second observation is that for each $t$ such that $t_{\mid_{\mathbf{a}^{\prime\prime}}} = t^{\prime\prime}$
the sum
$$\sum_{s_{\mid_{t^{\prime\prime}(\mathbf{a}^{\prime\prime})}}= s^{\prime\prime}} w_s$$
is simply the $s^{\prime\prime}$ component of
$(r^{s_0}_{t(\mathbf{a})})_{\mid_{t^{\prime\prime}(\mathbf{a}^{\prime\prime})}}$.
As $r^{s_0}$ is a global $\mathbb{Z}$-linear section this is equal to the fixed
parameter $w_{s^{\prime\prime}}$. So the sum in question reduces to
$$ w_{s^{\prime\prime}} \cdot \left( \sum_{t_{\mid_{\mathbf{a}^{\prime\prime}}} = t^{\prime\prime}} z_t  \right)  $$

The final observation, is that the remaining sum is the $t^{\prime\prime}$ component
of $(r^{t_0}_{\mathbf{a}})_{\mid_{\mathbf{a}^{\prime\prime}}}$ which, as $r^{t_0}$
is a global $\mathbb{Z}$-linear section, is equal to $r^{t_0}_{t^{\prime\prime}}$. This
gives the final form of the expression for $(r^{p_0}_{\mathbf{a}})_{\mid_{\mathbf{a}^{\prime\prime}}}$
as
$$   \sum_{t^{\prime\prime}, s^{\prime\prime}} z_{t^{\prime\prime}}w_{s^{\prime\prime}}(t^{\prime\prime}\circ s^{\prime\prime} ) $$
It is easy to see that the same arguments apply to $r^{p_0}_{\mathbf{a}^{\prime}}$ and
so
$$ (r^{p_0}_{\mathbf{a}})_{\mid_{\mathbf{a}^{\prime\prime}}} = (r^{p_0}_{\mathbf{a}^{\prime}})_{\mid_{\mathbf{a}^{\prime\prime}}}$$
as required.
\end{proof}

\section{Proof of Theorem \ref{thm:rings}}\label{sec:app2}

To prove this theorem we invoke a result from \cite{abramsky_et_al}
which considers a similar set-up to that seen in the previous sections and proves
a result relating the non-existence of solutions to a system of linear equations over
a ring $R$ to the non-triviality of a family of cohomological ``obstructions''.
We will recall their set-up, the relevant result and a characterisation of these
cohomological ``obstructions'' in terms of global $\mathbb{Z}$-linear sections before
proving Theorem \ref{thm:rings}.

\subsection{Result from \textit{Contextuality, cohomology \& paradox}}

In order to state the relevant theorem, we start with some preliminary definitions.
Let a \emph{ring-valued measurement scenario} be a triple $\langle X, \mathcal{M}, R \rangle$
where $X$ is a finite set, $\mathcal{M}$ is a downward closed cover of $X$ and $R$ is a ring.
An \emph{$R$-linear equation} on $\langle X, \mathcal{M}, R \rangle$ is a triple
$\phi = (V_{\phi}, a, b)$ where $V_{\phi} \in \mathcal{M}$, $a: V_{\phi}\rightarrow R$
and $b \in R$. Then for any $s \in R^{V_{\phi}}$ we say that $s \models \phi$ if
$$ \sum_{m \in V_{\phi}} a(m)s(m) = b   $$
 in the ring $R$. \\

An \emph{empirical model} $S$ on $\langle X, \mathcal{M}, R \rangle$ is a collection
of sets $\{ S_C \}_{C \in \mathcal{M}}$ where for each $C$, $S_C \subset R^C$ satisfying
the following compatibility condition for all $C, C' \in \mathcal{M}$
$$ \{s_{\mid_{C\cap C'}} \ \mid \ s \in S_C \} = \{s'_{\mid_{C\cap C'}} \ \mid \ s' \in S_{C'} \}  $$
We make the following observation linking relational structures over signatures
$\sigma \subset \sigma_R$ and empirical models which will be useful later.
\begin{obs}\label{obs:emp_model}
For any $\mathbf{CSP}(A,R)$ and $S \subset \mathbf{Hom}_k(A,R)$ which is non-empty,
and downward-closed and satisfies the forth property
 then the local sections of $S$ form an
empirical model for the measurement scenario $\langle A, A^{\leq k}, R \rangle$.
\end{obs}

For an empirical model $S$ on an $R$-valued measurement scenario, the \emph{$R$-linear theory}
of $S$ is the set of $R$-linear equations
$$ \mathbb{T}_R(S) = \{ \phi \ \mid \ \forall s \in S_{V_{\phi}}, \ s \models \phi \} $$

If $\mathbb{T}_R(S)$ is inconsistent (i.e.\ there is no $R$-assignment to all the
variables in $X$ simultaneously satisfying each of the $R$-linear equations in
the theory), then the empirical model $S$ is said to be
``all-vs-nothing for $R$'', written $\mathbf{AvN}_{R}(S)$.

We can now state the following results that we need for Theorem \ref{thm:rings}. The first result shows an important implication about the
cohomological obstructions in an empirical model which has an inconsistent $R$-linear theory.
\begin{theor}[Abramsky, Barbosa, Kishida, Lal, Mansfield \cite{abramsky_et_al}]\label{thm:avn}
For any ring $R$ and any $R$-valued measurement scenario $\langle X, \mathcal{M}, R \rangle$
and any empirical model $S$ we have that
$$ \mathbf{AvN}_{R}(S) \implies \mathbf{CSC}_{\mathbb{Z}}(S)$$
where $\mathbf{CSC}_{\mathbb{Z}}(S)$ means that for every local section $s$ in $S$ the
``cohomological obstruction'' of Abramsky, Barbosa and Mansfield $\gamma(s)$ is
non-zero.
\end{theor}

Next we have a result due to Abramsky, Barbosa and Mansfield which establishes
this useful equivalent condition for $\mathbf{CSC}_{\mathbb{Z}}(S)$
\begin{theor}[Abramsky, Barbosa, Mansfield \cite{ab_man_barb}] \label{thm:Z_ext}
  For any empirical model $S$, $\mathbf{CSC}_{\mathbb{Z}}(S)$ if and only if
  for every $s \in S_C$ there is no collection $\{r_{C'} \in \mathbb{Z}S_{C'} \}_{C' \in \mathcal{M}}$
  such that $r_C = s$ and for all $C_1, C_2 \in \mathcal{M}$
  $$ r_{C_1 \mid_{C_1 \cap C_2}} = r_{C_2 \mid_{C_1 \cap C_2}} $$
\end{theor}

This condition is precisely what inspired the cohomological $k$-consistency algorithm
and in the next section we show how these two results imply Theorem \ref{thm:rings}.

\subsection{Proof of Theorem \ref{thm:rings}}
We now prove the following equivalent formulation of Theorem \ref{thm:rings} which
replaces a structure with a ring representation with the underlying ring $R$
``represented as a relational structure''. This means simply that the relational
symbols (which are affine subsets of $R$ under the ring $R$) are labelled as
$E^{m}_{\mathbf{a}, b}$ for each $\mathbf{a}$ an $m$-tuple of elements of the ring $R$ and $b$ an element
of $R$ such that $(E^{m}_{\mathbf{a}, b})^R = \{ (r_1, \ldots , r_m) \ \mid \ \sum_{i} a_i\cdot r_i = b \}$.

\begin{theor}\label{thm:rings2}
For any finite ring $R$ represented as a relational structure over a finite signature $\sigma$,
there is a $k$ such that the cohomological $k$-consistency algorithm decides $\mathbf{CSP}(R)$.\\
Alternatively, there exists a $k$ such that for all $\sigma$-structures $A$
$$ A \rightarrow^{\mathbb{Z}}_{k} R \iff A \rightarrow R  $$
\end{theor}

\begin{proof}
The direction $A \rightarrow R \implies A \rightarrow^{\mathbb{Z}}_k R$ is easy
and is true for all signatures $\sigma$ and all $k\leq |A|$. Indeed note that to any
homomorphism $f : A \rightarrow R$ we can associate the  set $S_f = \{f_{\mid_{\mathbf{a}}} \}_{\mathbf{a}\in A^{\leq k}} \subset \mathbf{Hom}_k(A,R)$. It is not hard to see that $S_f$ is downward closed, has the forth
property and that $S_f$ is itself a global section witnessing the $\mathbb{Z}$-extendability
of each $f_{\mid_{\mathbf{a}}} \in S_f$. By Observation \ref{obs:consist_set}, this
implies that $A \rightarrow^{\mathbb{Z}}_{k} R$. \\

This leaves the more challenging direction, that there exists a $k$ such that
$A \not\rightarrow R \implies A \not\rightarrow^{\mathbb{Z}}_k R$ for all $A$. Suppose
that the maximum arity of a relation in $\sigma$ is $n$. Then as $R$ is a relational
model of a finite ring we know that each relation on $R$ is of the form
$E^{m}_{\mathbf{a}, b} = \{ (r_1, \ldots , r_m) \ \mid \ \sum_{i} a_i\cdot r_i = b \}$
where $\mathbf{a}$ is an $m$-tuple of elements of the ring $R$ and $b$ is an element
of $R$. We show that $k = n$ will suffice to identify all unsatisfiable instances $A$.

For $R$ and $\sigma$ as above any instance $\mathbf{CSP}(A,R)$ is specified by a set $A$
of variables where each related tuple $(x_1, \ldots , x_m ) \in (E^m_{\mathbf{a}, b})^A$
specifies an $R$-linear equation $\sum_{i} a_i\cdot x_i = b$. Call the collection
of such equations $\mathbb{T}^A$.
The fact that there is no homomorphism $A \rightarrow R$ is exactly the statement that $\mathbb{T}^A$
is unsatisfiable. Taking $k = n$, we have that the $R$-linear theory
 $\mathbb{T}_R(\mathbf{Hom}_k(A,R))$ (as defined in the previous section) contains
 $\mathbb{T}^A$ and so is unsatisfiable. We now show how this is sufficient to prove the theorem.

Consider running the cohomological $k$-consistency algorithm on the pair $(A, R)$
we get $S_0 = \overline{\mathbf{Hom}_k(A, R)}$. If $S_0 = \emptyset$ we are done.
Otherwise, by Observation \ref{obs:emp_model}, $S_0$ can be considered as an empirical
model on the measurement scenario $\langle A, A^{\leq k}, R \rangle$. Furthermore,
as $S_0 \subset \mathbf{Hom}_k(A, R)$, we have that $\mathbb{T}_R(S_0) \supset \mathbb{T}_R(\mathbf{Hom}_k(A,R))$.
This means in particular that $\mathbb{T}_R(S_0)$ is unsatisfiable by the assumption
that $A \not\rightarrow R$. By Theorems \ref{thm:avn} and \ref{thm:Z_ext}, this
means that no local section $s$ of $S_0$ is $\mathbb{Z}$-extendable in $S_0$, so
$S_1 = \emptyset$. So the cohomological $k$-consistency algorithm
rejects $(A, R)$ and $A \not\rightarrow^{\mathbb{Z}}_k R$, as required.

\end{proof}

It is notable that in the proof of this theorem, we see that the cohomological $k$-consistency
algorithm decides unsatisfiability of these systems of equations after just one iteration of its loop.
A future version of this work will investigate whether multiple iterations are required in
over different CSP domains. For now, we retain the iterative nature of the algorithm to
guarantee the conclusion in Observation \ref{obs:consist_set}.

\section{The strength of cohomological $k$-Weisfeiler-Leman}\label{sec:app3}

%
%
%
%
%
In this appendix, we demonstrate the power of $\equiv^Z_k$ to distinguish structures
which disagree on the CFI property, proving Theorem \ref{thm:main}.
To do this we give an equivalent definition of the cohomological $k$-Weisfeiler-Leman algorithm
and prove that this behaves well with appropriate logical interpretations.

\subsection{Cohomological $k$-Weisfeiler-Leman Equivalence}

The following is an alternative way of computing the $\equiv^{\mathbb{Z}}_k$ relation defined in the main article. Begin by computing
$S_0 = \overline{\overline{\mathbf{Isom}_k(A,B)}}$
as in the $k$-WL equivalence algorithm. If $S_0 = \emptyset$, then reject the pair
$(A,B)$ and halt. Otherwise we enter the following loop with $i=0$:
\begin{enumerate}
  \item Compute $S^{\mathbb{Z}}_i = \{s \in S_i \ \mid \ s \text{ is $\mathbb{Z}$-bi-extendable in }S_i \}$
  \item Compute $S_{i+1} = \overline{\overline{S^{\mathbb{Z}}_i}}$
  \item If $S_{i+1} = \emptyset$, then reject $(A,B)$ and halt
  \item If $S_{i+1} = S_i$ then accept $(A,B)$ and halt.
  \item Return to Step 1 with $i = i+1$.
\end{enumerate}
 If this algorithm accepts a pair $(A,B)$ we say that $A$ and $B$ are
 cohomologically $k$-equivalent and we write $A \equiv^{\mathbb{Z}}_k B$. \\

 We now record some simple facts about this equivalence. Firstly, by definition,
 this generalises $k$-equivalence and so $(k)$-WL equivalence, i.e.\
 $$A \equiv^{\mathbb{Z}}_k B \implies A \equiv_k B \iff A \equiv_{(k-1)-WL} B$$
Secondly, this algorithm determines a maximal set $S \subset \mathbf{Isom}_k(A,B)$
which is downward-closed, has the bijective forth property and for which each $f \in S$
is  $\mathbb{Z}$-extendable in $S$ and $f^{-1}$ is $\mathbb{Z}$-extendable in $S^{-1}$.
However, analogously to Observation \ref{obs:consist_set}, we note that the existence
of any non-empty $S$ satisfying these properties is a witness of $\equiv^{\mathbb{Z}}_k$.
\begin{obs}\label{obs:set_condition}
  For any two structures $A$ and $B$, $A \equiv^{\mathbb{Z}}_k B$ if and only if
  there exists a subset $S \subset \mathbf{Isom}_k(A,B)$ such that both $S$ and $S^{-1}$ are downward-closed,
  has the bijective forth property and have $\mathbb{Z}$-extendability for each of their elements.
\end{obs}
Finally, we observe that such a set also satisfies the conditions for witnessing
cohomological $k$-consistency of $\mathbf{CSP}(A, B)$ and $\mathbf{CSP}(B, A)$.
Formally we have
\begin{obs}
  For any two structures $A$ and $B$, $A \equiv^{\mathbb{Z}}_k B$ implies that
  $A \rightarrow^{\mathbb{Z}}_k B$ and $B \rightarrow^{\mathbb{Z}}_k A$.
\end{obs}
In the next section we establish how this equivalence relation behaves with respect
to logical interpretations.

\subsection{$\equiv^{\mathbb{Z}}_k$ and interpretations}

There are many different notions of logical interpretation in finite model theory.
The one we consider is defined as follows. A \emph{$\mathcal{C}^l$-interpretation}
$\Phi$ (of order $n$) of signature $\tau$
in signature $\sigma$ is a tuple of $\mathcal{C}^l[\sigma]$ formulas
$\langle \phi_R \rangle_{R \in \tau}$. For each relation symbol $R\in \tau$
of arity $r$, the formula $\phi_R$ has $nr$ free variables and is written as
$\phi_R(\mathbf{x}_1, \ldots , \mathbf{x}_r)$, where the $\mathbf{x}_i$ are $n$-tuples
of variables. Such an interpretation defines a map from $\sigma$-structures to
$\tau$-structures as follows. For any A, $\Phi(A)$ has universe $A^{n}$ and for
each relational symbol $R \in \tau$, the set of related tuples is given by
$$R^{\Phi(A)} := \{(\mathbf{a}_1, \ldots , \mathbf{a}_r) \in (A^n)^r \ \mid \ A, \mathbf{a}_1, \ldots , \mathbf{a}_r \models \phi_R  \}  $$

In the next result, we show that the equivalence $\equiv^{\mathbb{Z}}_k$ is preserved
by $C^l$-interpretations in the following way.

\begin{prop}\label{prop:Z_interpretation}
 For any (finite, relational) signatures $\sigma$ and $\tau$, $\sigma$-structures $A$ and $B$,
 natural numbers $n$ and $k$, and any order $n$ $\mathcal{C}^{nk}$-interpretation
 $\Phi$ of $\tau$ in $\sigma$ we have that
 $$ A \equiv^{\mathbb{Z}}_{nk} B \implies \Phi(A) \equiv^{\mathbb{Z}}_{k} \Phi(B)  $$

\end{prop}

\begin{proof}
By Observation \ref{obs:set_condition}, it suffices to show that there is a set
$S' \subset \mathbf{Isom}_k(\Phi(A),\Phi(B))$ which is downward-closed, satisfies
the bijective forth property and in which every map is $\mathbb{Z}$-extendable.
As $A \equiv^{\mathbb{Z}}_{nk} B$, there is already a set $S \subset \mathbf{Isom}_{nk}(A,B)$
satisfying these properties. For any $Q \subset A$ we use $S_Q$ to mean the elements
of $S$ with domain $Q$. We now show how to construct a suitable $S'$ from $S$.\\

For any $C \subset \Phi(A)$, let $\mathbf{\pi}(C)$ be the set of element in $A$
which appear in some tuple of $C$. As elements of $\Phi(A)$ are $n$-tuples over $A$,
it is clear that $|\mathbf{\pi}(C)| \leq n|C|$. We can now define $S'_{C}$ as the
 set of partial isomorphisms in $S_{\mathbf{\pi}(C)}$ applied coordinatewise to $C$, namely,
$$ \{ (f, \ldots , f)_{\mid_C} \ \mid \ f \in S_{\mathbf{\pi}(C)}  \}  $$
This is well defined for all $C \in (\Phi(A))^{\leq k}$ as $|\mathbf{\pi}(C)| \leq nk$.
That these maps define partial isomorphisms between $\Phi(A)$ and $\Phi(B)$ follows
from Hella's Lemma 5.1 in \cite{lhipt} which states that the elements of $\overline{\overline{\mathbf{Isom}_{nk}(A,B)}}$
are exactly those which preserve and reflect $C^{nk}$ formulas. As the relations
on $\Phi(A)$ and $\Phi(B)$ are defined by $C^{nk}$ formulas they are preserved and
reflected by the members of $S$. We now show that
$S' = \bigcup_{C \in \Phi(A)^{\leq k}} S'_C$ satisfies the required
properties. \\

\textbf{Downward-closure} This follows easily from  downward-closure of $S$. Suppose
$\mathbf{f} = (f, \ldots , f)_{\mid_{C}} \in S'$ and $\mathbf{g} \leq \mathbf{f}$. Then
there is some $C' \subset C$ such that $\mathbf{g} =  \mathbf{f}_{\mid_{C'}}$ and
$\mathbf{g} = (f_{\mid_{\mathbf{\pi}(C')}}, \ldots , f_{\mid_{\mathbf{\pi}(C')}})_{\mid_{C'}}$.
but $f_{\mid_{\mathbf{\pi}(C')}} \leq f$ and so is an element of $S$.\\

\textbf{Bijective forth property}
Let $\mathbf{f}\in S'_C$ with $|C| < k$, with $\mathbf{f}$ given as the coordinatewise
application of some $f \in S_{\mathbf{\pi}(C)}$. To show that $S'$ has the bijective
forth property we must show that there is a bijection $b : \Phi(A) \rightarrow \Phi(B)$
such that for any $\mathbf{a}\in \Phi(A)$ the function $\mathbf{f} \cup \{ (\mathbf{a}, b(\mathbf{a})) \} $
is in $S'_{C\cup\{\mathbf{a}\}}$. For any such $\mathbf{f}$, we can construct
a bijection $b$ whose image on any $\mathbf{a} \in \Phi(A)$ is given as
$$b(\mathbf{a}) = ( b^{\epsilon}(a_1), b^{\mathbf{a}_1}(a_2), \ldots , b^{(\mathbf{a}_{n-1})}(a_n) ) $$
where $\mathbf{a}_i$ is the $i$-tuple of the first $i$ elements in $\mathbf{a}$
and each $b^{\mathbf{a}_i}$ is a bijection $A\rightarrow B$. For any $\mathbf{a}\in \Phi(A)$
we choose the bijections $b^{\mathbf{a}_i}$ using the bijective forth property on
$S$. As $\mathbf{f}$ is a coordinatewise application of some $f \in S_{\mathbf{\pi}(C)}$
and as $|C|<k$ implies $|\mathbf{\pi}(C)| \leq nk - n < nk$, the bijective forth property for
$S$ implies the existence of a $b_1$ such that $f_1 = f \cup \{a_1, b_1(a_1)\} \in
S_{\mathbf{\pi}(C)\cup \{a_1\}}$. Let $b^{\epsilon} := b_1$. Now suppose for any $i < n$ we have
defined the bijections $b^{\epsilon}, b^{\mathbf{a}_1}, \ldots , b^{\mathbf{a}_i}$ and
$f_i = f \cup \{ (a_j, b^{\mathbf{a}_{j-1}}(a_j))\}_{1\leq j \leq i} \in S_{\mathbf{\pi}(C)\cup \{a_1, \ldots , a_i\}}$.
We still have $|\mathbf{\pi}(C)\cup \{a_1, \ldots , a_i\}| < nk$ so can use the
bijective forth property on $S$ again to find a bijection $b^{\mathbf{a}_i}$
such that $f_{i+1} = f_i \cup \{(a_i, b_{\mathbf{a}_i}(a_i))\} \in S_{\mathbf{\pi}(C)\cup \{a_1, \ldots , a_{i+1}\}}$.
This inductive procedure defines all the required bijections and furthermore
shows that $\mathbf{f} \cup \{(\mathbf{a}, b(\mathbf{a})\}$ is the coordinatewise application
of some $f_n \in S_{\mathbf{\pi}(C\cup \{\mathbf{a}\} )}$. This means in particular that
$\mathbf{f} \cup \{(\mathbf{a}, b(\mathbf{a})\}$ is in $S'_{C\cup\{\mathbf{a}\}}$,
as required.

\textbf{$\mathbb{Z}$-extendability} Our choice of $S'$ makes $\mathbb{Z}$-extendability
rather easy. Indeed, we see that any $\mathbf{f} = (f, \ldots , f) \in S'_{C}$ is $\mathbb{Z}$-extendable
because the $\mathbb{Z}$-linear global section extending $f \in S_{\mathbf{\pi}(C)}$
given as $s_C = \sum_{g \in S_C} \alpha_g g$ can be lifted to a $\mathbb{Z}$-linear
extension of $\mathbf{f}$ by defining $s'_{C} = \sum_{g \in S_{\mathbf{\pi}(C)}} \alpha_g (g, \ldots , g)$.
The properties of being a $\mathbb{Z}$-linear extension follow from those properties on $S$.

\end{proof}

\subsection{Deciding the CFI property}

Cai, F\"urer and Immerman\cite{cfi} showed that there is a property of relational structures
which can be decided in polynomial time but which cannot be expressed in infinitary
first-order logic with counting quantifiers for any number of variables. This
construction essentially encodes certain systems of linear equations (over
$\mathbb{Z}_2$) on top of graphs in such a way that isomorphism of the constructed
structures is determined by checking solvability of the systems of equations. In their
seminal paper\cite{cfi}, Cai, F\"urer and Immerman show that the solvable and unsolvable
versions of their construction cannot be distinguished in fixed point logic with counting.
Adaptations of this construction, encoding equations over different finite fields
were used by Dawar, Gr\"adel and Pakusa to show that adding rank quantifiers over
each finite field added distinct expressive power to FPC and a version using
equations over the rings $\mathbb{Z}_{2^q}$ was used by Lichter\cite{lich} to separate rank logic from \texttt{PTIME}.

As cohomological $k$-consistency was shown in the previous section to simultaneously decide
solvability over any finite ring, it is natural to ask whether the related equivalence
$\equiv^{\mathbb{Z}}_k$ can decide these CFI properties which are not definable
in FPC, rank logic or linear algebraic logic. We show in this section that it can.

Following Lichter\cite{lich}, we define the general CFI construction $\mathbf{CFI}_{q}(G, g)$
for $q$ a prime power, $G = (G, <)$ an ordered undirected graph and $g$ a function from the edge
set of $G$ to $\mathbb{Z}_q$. The idea is that the construction encodes a system
of linear equations over $\mathbb{Z}_q$ into $G$ while the function $g$ ``twists'' these
equations in a certain way. For CFI structures, $\mathbf{CFI}_{q}(G, g)$ the property $\sum g = 0$ is sometimes
called the \emph{CFI property}. The following well-known fact (see \cite{pakusa}, for example)
shows that this property is closed under isomorphisms and is useful in our later arguments.
\begin{fact}
  For any prime power, $q$, ordered graph $G$, and functions $g, h$ from the edges
  of $G$ to $\mathbb{Z}_q$,
  $$ \mathbf{CFI}_{q}(G, g) \cong \mathbf{CFI}_{q}(G, h) \iff \sum g = \sum h $$
\end{fact}

$\mathbf{CFI}_{q}(G, g)$ is built in three steps. First, we define a gadget which replaces each vertex of $x$ with elements that form a ring. Secondly, we define relations between gadgets which impose consistency equations
between gadgets. Finally, the function $g$ is used to insert the important
twists into the consistency equations. We now describe this in detail below, following
a presentation by Lichter\cite{lich}.

\textbf{Vertex gadgets}
For any vertex $x \in G$, let $N(x)$ be the neighbourhood of $x$ in $G$ (i.e.\ those
vertices which share edges with $x$) and let $\mathbb{Z}_q^{N(x)}$ denote the ring of
functions from $N(x)$ to the ring $\mathbb{Z}_q$. We will replace each vertex $x$ of
the base graph with a gadget whose vertices are the following subset of $\mathbb{Z}_q^{N(x)}$,
$$A_x = \{\mathbf{a} \in \mathbb{Z}_q^{N(x)}  \ \mid \ \sum_{y\in N(x)} \mathbf{a}(y) = 0 \} $$
The relations on the gadget are for each $y$ in $N(x)$ a symmetric relation
$$I_{x,y} = \{(\mathbf{a},\mathbf{b}) \ \mid \ \mathbf{a}(y) = \mathbf{b}(y)\} $$
and a directed cycle encoded by the relation
$$C_{x,y} = \{(\mathbf{a},\mathbf{b}) \ \mid \ \mathbf{a}(y) = \mathbf{b}(y) + 1\} $$
Together these impose the ring structure of $\mathbb{Z}_q^{N(x)}$ onto the
vertices of the gadget.

\textbf{Edge equations}
Next define a relation between gadgets for each edge $\{x,y \}$ in G and each
constant $c \in \mathbb{Z}_q$ of the form
$$E_{\{x,y\}, c}  = \{ (\mathbf{a}, \mathbf{b}) \ \mid \ \mathbf{a}
\in A_x,\ \mathbf{b} \in A_y,\ \mathbf{a}(y)+ \mathbf{b}(x) = c  \} $$

\textbf{Putting it together with a twist}
We finally define the structure $\mathbf{CFI}_q(G,g)$ as
$\langle A, \prec, R_I, R_C, R_{E,0}, R_{E,1}, \ldots , R_{E,q-1}  \rangle$ where
the universe is $A = \cup_{x} A_x$ where $\prec$ is the linear pre-order
$$ \prec = \bigcup_{x < y} A_x \times A_y $$
and the edge equations $R_{E,c}$ are interpreted according to the twists in $g$ as
$$ R_{E,c} = \bigcup_{e \in E} E_{e, c+g(e)} $$
where the sum in the subscript is over $\mathbb{Z}_q$
For the relations $R_I$ and $R_C$ we deviate slightly from Lichter's
construction and interpret these as ternary relations of the following form
\begin{align*}
  R_I &= \bigcup_{\{x,y\} \in E } I_{x,y} \times A_y \\
  R_C &= \bigcup_{\{x,y\} \in E } C_{x,y} \times A_y
\end{align*}

We now use recall the two major separation results based on this construction. The
first is a landmark result of descriptive complexity from the early 1990's.

\begin{theor}[Cai, Furer, Immerman\cite{cfi}]
  There is a class of ordered (3-regular) graphs $\mathcal{G} = \{ G_n \}_{n\in \mathbb{N}}$
  such that in the respective class of CFI structures
  $$\mathcal{K} = \{ \mathbf{CFI}_2(G,g) \ \mid \ G \in \mathcal{G}   \} $$
  the CFI property is decidable in polynomial-time but cannot be
  expressed in FPC.
\end{theor}
The second is a recent breakthrough due to Moritz Lichter.
\begin{theor}[Lichter\cite{lich}]
There is a class of ordered graphs $\mathcal{G} = \{ G_n \}_{n\in \mathbb{N}}$
such that in the respective class of CFI structures
$$\mathcal{K} = \{ \mathbf{CFI}_{2^k}(G,g) \ \mid \ G \in \mathcal{G}  \} $$
the CFI property is decidable in polynomial-time (indeed, expressible in
choiceless polynomial time) but cannot be expressed in rank logic.
\end{theor}

We now show that in both of these classes there exists a fixed $k$ such that
$\equiv^{\mathbb{Z}}_k$ distinguishes structures which differ on the CFI property. This relies
on two lemmas. The first shows that this property is equivalent to the solvability
of a certain system of equations over $\mathbb{Z}_q$, while the second shows
that this system of equations can be interpreted in on the classes above with a
uniform bound on the number of variables per equation.

The first lemma is an adaptation of Lemma 4.36 from Wied Pakusa's PhD thesis\cite{pakusa}.
We begin by defining for any $\mathbf{CFI}q(G,g)$ a system of linear equations
over $\mathbb{Z}_q$. This system, $\mathbf{Eq}_q(G,g)$, is the following
collection of equations:
\begin{itemize}
  \item $X_{\mathbf{a},u}$ for all $u \in G$ and all $\mathbf{a} \in A_u \subset \mathbf{CFI}_q(G,g)$,
  \item $I_{\mathbf{a},\mathbf{b},v}$ for all $u \in G$ and $\mathbf{a}, \mathbf{b} \in A_u$ such
  that there exists $v \in N(u)$ and $\mathbf{c}\in A_v$ such that $(\mathbf{a}, \mathbf{b}, \mathbf{c}) \in R_I$,
  \item $C_{\mathbf{a},\mathbf{b},v}$ for all $u \in G$ and $\mathbf{a}, \mathbf{b} \in A_u$ such
  that there exists $v \in N(u)$ and $\mathbf{c}\in A_v$ $(\mathbf{a}, \mathbf{b}, \mathbf{c}) \in R_C$, and
  \item $E_{\mathbf{a}, \mathbf{b},c}$ for all $\mathbf{a}\in A_u, \mathbf{b} \in A_{v}$ and
  $(\mathbf{a}, \mathbf{b})\in R_{E,c}$
\end{itemize}
where the variables are $w_{\mathbf{a},v}$ for every $u \in G,$ $\mathbf{a} \in A_u$ and
$v \in N(u)$ and the equations are given as:
\begin{align*}
  X_{\mathbf{a},u} : \quad \sum_{v\in N(u)} w_{\mathbf{a},v} &= 0 \\
  I_{\mathbf{a},\mathbf{b},v} : \quad w_{\mathbf{a},v} - w_{\mathbf{b},v} &= 0 \\
  C_{\mathbf{a},\mathbf{b},v} : \quad w_{\mathbf{a},v} - w_{\mathbf{b},v}  &= 1 \\
  E_{\mathbf{a}, \mathbf{b},c} : \quad w_{\mathbf{a},v} + w_{\mathbf{b},u}  &=  c
\end{align*}

Then we have the following lemma.

\begin{lem}\label{lem:equation}
  $\mathbf{CFI}_q(G,g)$ a CFI structure, has $\sum g = 0$ if and only if $\mathbf{Eq}_q(G,g)$
  is solvable in $\mathbb{Z}_q$
\end{lem}

\begin{proof}
 Firstly we recall Fact \ref{fact:cfi} that $\sum g = 0$ if and only if there is
 an isomorphism $f: \mathbf{CFI}_q(G,g) \rightarrow \mathbf{CFI}_q(G,\mathbf{0})$,
 where $\mathbf{0}$ is the constant $0$ function. We now show that there is such an
 isomorphism if and only if there is a solution to $\mathbf{Eq}_q(G,g)$.

 For the forward direction, suppose that we have an isomorphism $f: \mathbf{CFI}_q(G,g) \rightarrow \mathbf{CFI}_q(G,\mathbf{0})$. Now as $f$ is a bijection and preserves the
 pre-order $\prec$, we have that for any $u \in G$, $f$ maps $A_u$ to $A_u$. This means
 that for any $\mathbf{a} \in A_u$ $f(\mathbf{a})$ is a function in $\mathbb{Z}_q^{N(u)}$.
This means that the assignment $w_{\mathbf{a}, v} \mapsto f(\mathbf{a})(v)$ is well-defined
for all the variables in  $\mathbf{Eq}_q(G,g)$. We now show that this assigment
satisfies the system of equations. The $X$ equations in $\mathbf{Eq}_q(G,g)$ become
the statement that for all $u \in G$ and $\mathbf{a} \in A_u$
$$\sum_{v \in N(u)} f(\mathbf{a})(v) = 0 $$
which follows directly from the fact that $f(\mathbf{a}) \in A_u$. For the $I$ and $C$ equations,
we note that as $f$ preserves all relations from $\mathbf{CFI}_q(G,g)$. So for any $\mathbf{a}, \mathbf{b} \in A_u$
and $\mathbf{c} \in A_v$ such that $(\mathbf{a}, \mathbf{b}, \mathbf{c})$ is related by $R_I$ or
$R_C$ in $\mathbf{CFI}_q(G,g)$ then $(f(\mathbf{a}), f(\mathbf{b}), f(\mathbf{c}))$
is similarly related in $\mathbf{CFI}_q(G,\mathbf{0})$. The definitions of these relations
imply that $f(\mathbf{a})(v) - f(\mathbf{b})(v)$ is $0$ or $1$ respectively, which implies
that our assignment to the variables $w_{\mathbf{a},v}$ and $w_{\mathbf{b},v}$ satisfies the
relevant $I$ or $C$ equation. A similar argument applies to the $E$ equations except that
the conclusion from $(f(\mathbf{a}), f(\mathbf{b})) \in R_{E, c}$ in $\mathbf{CFI}_q(G,\mathbf{0})$ that the
relevant $E$ equation is satisfied follows from the fact that there is no twisting
of the $R_{E,c}$ relation in $\mathbf{CFI}_q(G,\mathbf{0})$.

The reverse direction is the observation that any satisfying assignment
to the variables $w_{\mathbf{a},v}$ in $\mathbf{Eq}_q(G,g)$ defines an isomorphism from $\mathbf{CFI}_q(G,g)$
to $\mathbf{CFI}_q(G,\mathbf{0})$ where $f(\mathbf{a})(v) = w_{\mathbf{a},v}$.
Satisfying the $X$ equation guarantees that for $\mathbf{a} \in A_u$ its image
$f(\mathbf{a})$ is also in $A_u$. Satisfying the $I$ and $C$ equations ensures that the $R_I$
and $R_C$ relations are preserved. So,
the additive structure of $\mathbb{Z}_q^{N(u)}$ is preserved in $A_u$ and thus
$f$ is bijective. Finally the $E$ equations define the $R_{E,c}$ relation in
$\mathbf{CFI}_q(G,\mathbf{0})$ and so satisfying these ensures that $f$ preserves the
$R_{E,c}$ relation.
\end{proof}

It is not hard to see that the system $\mathbf{Eq}_q(G,g)$ is first order interpretable
in $\mathbf{CFI}_q(G,g)$. However, Theorem \ref{thm:rings} shows that cohomological
$k$-consistency decides satisfiability of systems of equations over
any ring in with up to $k$ variables per equation. Thus to show that cohomological
$k$-equivalence distinguishes positive and negative instances of the CFI property for some fixed $k$ we need to show that
an equivalent system of equations can be interpreted which fixes the number of variables per equation.
This is the content of the following lemma.

\begin{lem}\label{lem:interp}
  For any prime power $q$, there is an interpretation $\Phi_q$ from the signature of the
  CFI structures $\mathbf{CFI}_q(G, g)$ to the signature of the ring $\mathbb{Z}_q$
  with relations of arity at most $3$ such that
      $$ \Phi_q(\mathbf{CFI}_q(G,g)) \rightarrow \mathbb{Z}_q \iff \sum g = 0 $$
\end{lem}

\begin{proof}
  From Lemma \ref{lem:equation}, we know that interpreting the system of equations
  $\mathbf{Eq}_q(G,g)$ would suffice for this purpose. However, the $X$ equations
  in $\mathbf{Eq}_q(G,g)$ contain a number of variables which grows with the
  size of the maximum degree of a vertex in $G$. As this is, in general, unbounded
  - and in particular is unbounded in Lichter's class - we need to introduce some
  equivalent equations in a bounded number of variables. To do this we will introduce
  some slack variables and utilise the ordering on $G$ to turn any such equation
  in $n$ variables into a series of equations in $3$ variables. We now describe the
  interpretation $\Phi_q$ as follows.\\
  Let $3$-$\mathbb{Z}_q$ denote the relational structure which contains a relation
  $T_{\mathbf{\alpha}, \beta}$ for each $\mathbf{\alpha}$ a tuple of elements of
  $\mathbb{Z}_q$ size up to 3 and $\beta \in \mathbb{Z}_q$. Each related tuple
  $(x,y,z) \in T_{\mathbf{\alpha}, \beta}$ in a $3$-$\mathbb{Z}_q$ structure is an
  equation
  $$ \alpha_1 x + \alpha_2 y + \alpha_3 z = \beta $$
  To help define the interpretation we introduce some shorthand for some easily interpretable
  relations on CFI structures $A$. For $\mathbf{a}, \mathbf{b} \in A$ write
  $\mathbf{a} \sim \mathbf{b}$ if the two elements belong to the same gadget in
  $A$ and $\mathbf{a} \frown \mathbf{b}$ if they belong to adjacent gadgets.
  Both of these relations are easily first-order definable as $\mathbf{a} \sim \mathbf{b}$
  if and only if they are incomparable in the $\prec$ relation and $\mathbf{a}\frown \mathbf{b}$
  if and only if $(\mathbf{a}, \mathbf{b})\in R_{E,c}$ for some $c$.
For $\mathbf{a} \frown \mathbf{b}$ in $A$ we will refer to the elements $(\mathbf{a}, \mathbf{a},\mathbf{b})$ and
  $(\mathbf{a}, \mathbf{b},\mathbf{b})$ as $w_{\mathbf{a},\mathbf{b}}$ and $z_{\mathbf{a}, \mathbf{b}}$.
  These will be the variables in the interpreted system of equations. As $A$ comes with a linear pre-order $\prec$ inherited from the order on $G$, we
  can also define a local predecessor relation in the neighbourhood of any $\mathbf{a} \in A$.
  We say that $\mathbf{b}$ is a local predecessor of $\mathbf{b}^{\prime}$ at $\mathbf{a}$
  and write $\mathbf{b} \vdash_{\mathbf{a}} \mathbf{b}^{\prime}$ if $\mathbf{a}\frown \mathbf{b}$
  and $\mathbf{a} \frown \mathbf{b}^{\prime}$ and there is no $\mathbf{b}^{\prime\prime}$
  with $\mathbf{a} \frown \mathbf{b}^{\prime\prime}$ such that $\mathbf{b} \prec
  \mathbf{b}^{\prime\prime} \prec \mathbf{b}^{\prime}$.\\

  Now we define the interpretation  on $A^{3}$ in three steps, resulting
  in a system of equations which is solvable if and only if $\mathbf{Eq}_q(G,g)$
  is solvable.
  \textbf{Step 1: Reducing variables}
  We note that in $\mathbf{Eq}_q(G,g)$ there are only variables $w_{\mathbf{a},y}$
  for $\mathbf{a} \in A_x$ and $y \in N(x)$, whereas the shorthand above describes
  variables $w_{\mathbf{a},\mathbf{b}}$ and $z_{\mathbf{a},\mathbf{b}}$ for all $\mathbf{a}\in A_x$ and $\mathbf{b}\in A_y$. To reduce the number of variables we want to interpret,
  for all $\mathbf{a}\frown\mathbf{b}$ and $\mathbf{b} \sim \mathbf{b}^{\prime}$,
  the equations $w_{\mathbf{a},\mathbf{b}} = w_{\mathbf{a},\mathbf{b}^{\prime}}$
  and $z_{\mathbf{a},\mathbf{b}} = z_{\mathbf{a},\mathbf{b}^{\prime}}$. This is done
  by add the pairs $(w_{\mathbf{a},\mathbf{b}},w_{\mathbf{a},\mathbf{b}^{\prime}})$ and $(z_{\mathbf{a},\mathbf{b}},z_{\mathbf{a},\mathbf{b}^{\prime}})$ to the relation  $T_{(1,-1),0}$
  which can be done as $\frown$ and $\sim$ are definable.

  \textbf{Step 2: Interpreting $I, C $ and $E$ equations}
   Defining these equations in $\Phi(A)$ is straightforward as they all have fewer
   than $3$ variables. In particular we want to add equations
   $$w_{\mathbf{a},\mathbf{b}} - w_{\mathbf{a}^{\prime},\mathbf{b}} = 0$$ for any
   $(\mathbf{a}, \mathbf{a}^{\prime}, \mathbf{b})\in R_{I}$,
   $$w_{\mathbf{a},\mathbf{b}} - w_{\mathbf{a}^{\prime},\mathbf{b}} = 1$$
   for any $(\mathbf{a}, \mathbf{a}^{\prime}, \mathbf{b})\in R_{C}$, and
   $$ w_{\mathbf{a},\mathbf{b}} + w_{\mathbf{b}, \mathbf{a}} = c $$
   for any $(\mathbf{a}, \mathbf{b}) \in R_{E,c}$. These are all easily first-order
   definable in the $\mathbf{CFI}_q$ signature.

  \textbf{Step 3: Interpreting $X$ equations}
  To interpret the equations for each $u\in G$ and $\mathbf{a} \in A_u$
  $$ \sum_{v \in N(u)} w_{\mathbf{a}, v} = 0 $$
  in $\Phi(A)$, we first note that the linear order on $G$ restricts to a
  linear order on $N(u)$ which we can write as $\{v_1, \ldots , v_n\}$ where $i<j$
  if and only if $v_i < v_j$. To do this it suffices to impose the equations
  $$ w_{\mathbf{a},\mathbf{b}_1} + \dots + w_{\mathbf{a},\mathbf{b}_n} = 0$$
  for each sequence of elements $\mathbf{b}_1 \vdash_{\mathbf{a}}  \ldots \vdash_{\mathbf{a}}
  \mathbf{b}_n$ with $\mathbf{b}_i \in A_{v_i}$. To do this in equations
  with at most three variables we employ the auxiliary $z$ variables in the following
  way. For any $\mathbf{a} \mathbf{b} \in A$ such that $\mathbf{a}\frown \mathbf{b}$,
  if there is no $\mathbf{b}^{\prime}$
  such that $\mathbf{b}^{\prime} \vdash_{\mathbf{a}} \mathbf{b}$, then we interpret the equation
  $$ w_{\mathbf{a}, \mathbf{b}} - z_{\mathbf{a},\mathbf{b}} = 0 $$
  if there is $\mathbf{b}^{\prime}$ such that $\mathbf{b}^{\prime} \vdash_{\mathbf{a}} \mathbf{b}$
  then interpret for all such $\mathbf{b}^{\prime}$ the equation
  $$ z_{\mathbf{a}, \mathbf{b}^{\prime}} + w_{\mathbf{a},\mathbf{b}} - z_{\mathbf{a}, \mathbf{b}} = 0 $$
  and if there is no $\mathbf{b}^{\prime}$ such that $\mathbf{b} \vdash_{\mathbf{a}} \mathbf{b}^{\prime}$
  then interpret the equation
  $$ z_{\mathbf{a}, \mathbf{b}} = 0 $$
   In this system of equations the $z_{\mathbf{a},\mathbf{b}}$ variables act as
   running totals for the sum $\sum w_{\mathbf{a}, \mathbf{b}_i}$ and so it is not
   hard to see that solutions to these equations are precisely solutions to the
   equations $\sum w_{\mathbf{a}, \mathbf{b}_i} = 0$. Furthermore, as the relation
   $\vdash_{\mathbf{a}}$ is definable in the signature of the $\mathbf{CFI}_q$
   structures so too are these equations.

   To conclude, we have interpreted in $\Phi(\mathbf{CFI}_q(G,g))$ a system of linear
   equations with three variables per equation which is solvable over $\mathbb{Z}_q$
   if and only if $\mathbf{Eq}_q(G,g)$ is solvable. Thus there is a homomorphism
   $\Phi(\mathbf{CFI}_q(G,g)) \rightarrow \mathbb{Z}_q$ (as $3$-$\mathbb{Z}_q$ structures)
   if and only if $\sum g = 0$.
\end{proof}

We can now conclude with the proof of Theorem \ref{thm:main}.

\begin{proof}[Proof of Theorem \ref{thm:main}]
  By Fact \ref{fact:cfi}, the reverse implication is easy as $\sum h = 0$ implies that
  $\mathbf{CFI}_{q}(G,g) \cong \mathbf{CFI}_{q}(G,h)$ and so the structures
  are cohomologically $k$-equivalent for any $k$. \\
  The converse follows from the series of lemmas we have just presented. If
  $\sum h \neq 0$ then by Lemma \ref{lem:interp} there is an interpretation $\Phi_q$
  of order $3$ such that $\Phi_q(\mathbf{CFI}_{q}(G,g)) \rightarrow \mathbb{Z}_q$
  but $\Phi_q(\mathbf{CFI}_{q}(G,h)) \not\rightarrow \mathbb{Z}_q$. By Theorem \ref{thm:rings},
  This is means that $\Phi_q(\mathbf{CFI}_{q}(G,g)) \rightarrow^{\mathbb{Z}}_3 \mathbb{Z}_q$
  but $\Phi_q(\mathbf{CFI}_{q}(G,h)) \not\rightarrow^{\mathbb{Z}}_3 \mathbb{Z}_q$. So by
  Observation \ref{obs:backforth}, we must have that $\Phi_q(\mathbf{CFI}_{q}(G,g)) \not\equiv^{\mathbb{Z}}_3 \Phi_q(\mathbf{CFI}_{q}(G,h))$. Then noting that the number of variables used in
  the interpretation $\Phi_q$ is some constant $c$ not depending on $q$ and assuming
  without loss of generality that $k$ is greater than $3c$ then Proposition \ref{prop:Z_interpretation}
  implies that $\mathbf{CFI}_{q}(G,g) \not\equiv^{\mathbb{Z}}_k \mathbf{CFI}_{q}(G,h)$, as required.
\end{proof}

\end{document}